\documentclass[onecolumn]{IEEEtran}
\usepackage{subfigure,cite,graphicx,amsmath,amssymb,mathrsfs,epsfig,dblfloatfix,epstopdf}
\usepackage{amsmath,amssymb,mathrsfs,graphicx}
\usepackage{psfrag}

\newtheorem{definition}{Definition}
\newtheorem{theorem}{Theorem}
\newtheorem{example}{Example}
\newtheorem{remark}{Remark}
\newtheorem{lemma}{Lemma}

\newtheorem{corollary}{Corollary}

\begin{document}
\title{Linear Degrees of Freedom of MIMO Broadcast Channels with Reconfigurable Antennas in the Absence of CSIT}

\author{Minho~Yang,~\IEEEmembership{Student~Member,~IEEE,}
        Sang-Woon~Jeon,~\IEEEmembership{Member,~IEEE,}
				and~Dong~Ku~Kim,~\IEEEmembership{Member,~IEEE}
\thanks{M. Yang and D. K. Kim were funded by the Ministry of Science, ICT $\&$ Future Planning (MSIP), Korea in the ICT R$\&$D Program 2014.
S.-W. Jeon was supported in part by the Basic Science Research Program through the National Research Foundation of Korea (NRF) funded by the Ministry of Education, Science and Technology (MEST) [NRF-2013R1A1A1064955].
}
\thanks{M. Yang and D. K. Kim are with the School of Electrical and Electronic Engineering,
Yonsei University, Seoul, South Korea (e-mail:
\{navigations, dkkim\}@yonsei.ac.kr).}
\thanks{S.-W. Jeon is with the Department of Information and Communication Engineering, Andong National University, Andong, South Korea (e-mail:
swjeon@anu.ac.kr).}
}
\maketitle

\IEEEpeerreviewmaketitle

\begin{abstract}
The $K$-user multiple-input and multiple-output (MIMO) broadcast channel (BC) with no channel state information at the transmitter (CSIT) is considered,
where each receiver is assumed to be equipped with reconfigurable antennas capable of choosing a subset of receiving modes from several preset modes.
Under general antenna configurations, the sum linear degrees of freedom (LDoF) of the $K$-user MIMO BC with reconfigurable antennas is completely characterized, which corresponds to the maximum  sum DoF achievable by linear coding strategies.  
The LDoF region is further characterized for a class of antenna configurations.
Similar analysis is extended to the $K$-user MIMO interference channels with reconfigurable antennas and the sum LDoF is characterized for a class of antenna configurations.
\end{abstract}


\section{Introduction}

Recently, there have been considerable researches on characterizing the \emph{degrees of freedom} (DoF) of wireless networks.
As current wireless networks become very complicated, exact capacity characterization is so difficult that many researchers have actively studied approximate capacity characterizations in the shape of DoF.
The DoF is the prelog factor of capacity, providing an intuitive metric for the number of interference-free communication channels that wireless networks can attain at the  high signal-to-noise ratio (SNR) regime.
Hence, it is regarded as a primary performance metric for multiantenna and/or multiuser communication systems.
Cadambe and Jafar recently made a remarkable progress on understanding DoF of multiuser wireless networks showing that the sum DoF of the $K$-user interference channel (IC) is given by $K/2$ \cite{Cadambe:08}.
An innovative methodology called \emph{interference alignment} (IA) has been proposed to obtain $K/2$ DoF, which aligns multiple interfering signals into the same signal space at each receiver. 
The concept of such signal space alignment has been successfully adapted to various network environments, e.g., see \cite{Viveck2:09, Viveck1:09,Tiangao:10, Suh:11,Tiangao:12,Jeon4:12,Jeon:14} and the references therein.
More recently, different strategies of IA were further developed in terms of ergodic IA \cite{Nazer11:09,Jeon5:13,Jeon2:11,Jeon2:14} and real IA \cite{Motahari:09,Motahari2:09}.

Note that most of the previous researches including the aforementioned IA techniques have focused on DoF of wireless networks under the assumption that each transmitter perfectly knows global channel state information (CSI).
However, for many practical communication systems, acquiring the exact CSI value at transmitters is very challenging due to channel feedback delay, system overhead, and so on.
Motivated by these practical restrictions, implementing IA under a more relaxed CSI condition has been actively studied in the literature.
Maddah-Ali and Tse made a breakthrough in \cite{Maddah-Ali:12} demonstrating that completely outdated  CSI is still useful to improve DoF of the $K$-user multiple-input and single-output (MISO) broadcast channel (BC).
Preceded by \cite{Maddah-Ali:12}, there have been a series of researches for studying IA techniques exploiting outdated or delayed CSI at transmitters \cite{Vaze:11, Abdoli:11, Vaze:12_dCSIT, Abdoli:11_IC,Abdoli:13}. 
In \cite{Vaze:11, Abdoli:11, Vaze:12_dCSIT}, similar DoF gains were shown in MIMO BC under delayed CSIT and, in particular, the DoF region of the two-user MIMO BC with delayed CSIT was completely characterized in \cite{Vaze:12_dCSIT}.
In the context of IC, it has been first shown in \cite{Maleki:12} that IA can achieve more than one DoF in the three-user SISO IC under delayed CSIT, which is then extended to the $K$-user case  in \cite{Abdoli:11_IC,Abdoli:13}.

Although there is still a practical demand for further relaxing CSI requirements at the transmitter side, it has been proved in \cite{Jafar:05} that the DoF of the $K$-user MISO BC collapses to one for isotropic fading if the transmitter cannot acquire any information about CSI.
In terms of isotropic fading and no CSIT, similar DoF degradation was further shown in MIMO BC and IC \cite{Vaze:12, Huang:12, Zhu:12, Vaze:12_IC}.
On the other hand, IA without CSIT, called \emph{blind IA}, has been recently proposed in \cite{Jafar:12} for a class of heterogeneous block fading models\footnote{Certain users experience smaller coherence time/bandwidth than others (See \cite{Jafar:12} for more details).}  
achieving larger DoF than that achievable for the isotropic fading model.
In addition, it was shown that blind IA obtains similar DoF gain for a class of homogeneous block fading models\footnote{All users experience independent block fading with the same coherence time, but different offsets (See \cite{Zhou:12} for more details).} \cite{Zhou:12, Zhou:12_dio, Zhou:12_Imp}.


In \cite{Gou:10, Gou:11}, Gou, Wang, and Jafar have first proposed a blind IA technique exploiting \emph{reconfigurable antennas}. 
As shown in Fig. \ref{reconfigurableantenna}, reconfigurable antennas are capable of dynamically adjusting their radiation patterns in a controlled and reversible manner through various technologies such as solid state switches or microelectromechanical switches (MEMS), which can be conceptually modeled as antenna selection that each RF-chain of reconfigurable antennas chooses one of receiving mode among several preset modes at each time instant, see also \cite[Section I]{Gou:11} for the concept of reconfigurable antennas.
Based on a remarkable observation that even for time-invariant channels, reconfigurable antenna can artificially create channel matrices correlated across time in some specific structure, 
the authors in \cite{Gou:11} show that the optimal sum DoF of 
the $K$-user $M \times 1$ MISO BC is given by $\frac{MK}{M+K-1}$ when each user is equipped with a reconfigurable antenna whose RF-chain can choose one receiving mode from $M$ preset modes.
Subsequently, in \cite{Wang:10}, the achievability result in \cite{Gou:11} is generalized to the $K$-user $M \times N$ MIMO BC where each user is equipped with a set of reconfigurable antennas whose RF-chains are able to choose $N$ receiving modes from $M$ preset modes, showing that the sum DoF of $\frac{MNK}{M+NK-N}$  is achievable.
The idea of blind IA using reconfigurable antennas is further extended to ICs consisting of receivers with reconfigurable antennas \cite{Wang:11, Wang:11_2, Lu:13, Wang:14, Lu:14}. 

In this paper, we consider the \emph{$K$-user MIMO BC assuming a general reconfigurable antenna environment}. 
In particular, the transmitter is equipped with $M$ antennas and user $k$, $k=1,\cdots,K$, is equipped with a set of reconfigurable antennas whose RF-chains can choose $L_{k}$ receiving modes from $N_{k}$ preset modes ($N_{k}\geq L_k$), which includes the conventional non-reconfigurable antenna model ($N_k=L_k$ for this case). 
We focus on the \emph{linear DoF (LDoF) with no CSIT}, i.e., the maximum DoF achievable by linear coding strategies with no CSIT, see also \cite{Lashgari:13, Lashgari:14, Kao:14} for the definition of LDoF. 
For general antenna configurations, we completely characterize the sum LDoF of the $K$-user MIMO BC with reconfigurable antennas in the absence of CSIT.
We further characterize the LDoF region for a specific class of antenna configurations.
Therefore, the main contributions of this paper are two-folds:
1) we generalize the previous achievability results in \cite{Gou:11, Wang:10} assuming a certain class of antenna configurations to general antenna configurations,
2) we show the converse of our achievable DoF in the LDoF sense, which implies that the achievability result in \cite{Wang:10} is also optimal in the LDoF sense.
Our analysis is further applied to a class of $K$-user MIMO IC with reconfigurable antennas and the sum LDoF is characterized for a class of   antenna configurations, which generalizes the achievable sum DoF result in \cite{Lu:13}.

The rest of this paper is organized as follows. In Section \ref{sec:system_model}, we introduce the $K$-user MIMO BC with reconfigurable antennas. 
In Section \ref{sec:main_results}, we first define the LDoF and state the main result of this paper, the sum LDoF and LDoF region of the $K$-user MIMO BC with reconfigurable antennas. 
We present the converse and achievability of the main results in Section \ref{sec:converse} and \ref{sec:achievability}, respectively and finally conclude in Section \ref{sec:conclusion}.

\section{System Model}
\label{sec:system_model}

\begin{figure}[!t]
\centering
\includegraphics[width=3.5in]{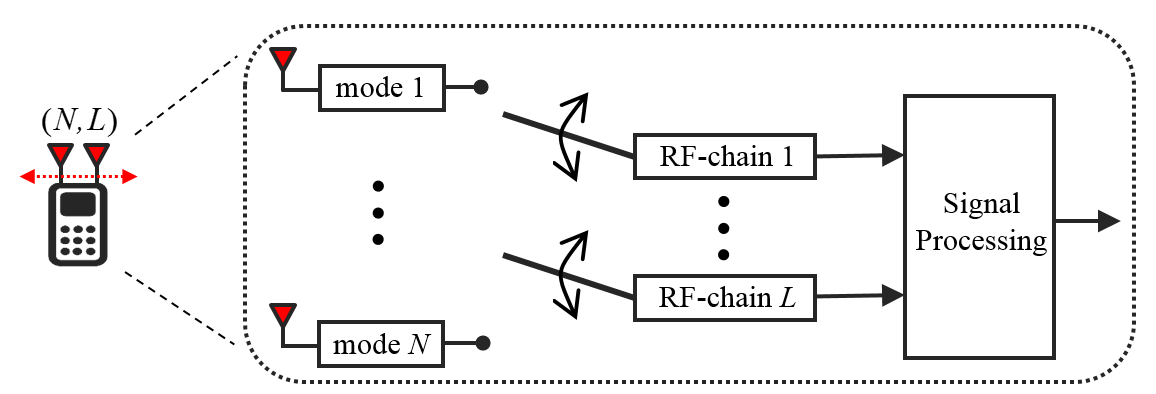}
\caption{Conceptual illustration of reconfigurable antennas, where each of $L$ RF-chains chooses one receiving mode from $N$ preset modes.}
\label{reconfigurableantenna}
\includegraphics[width=3in]{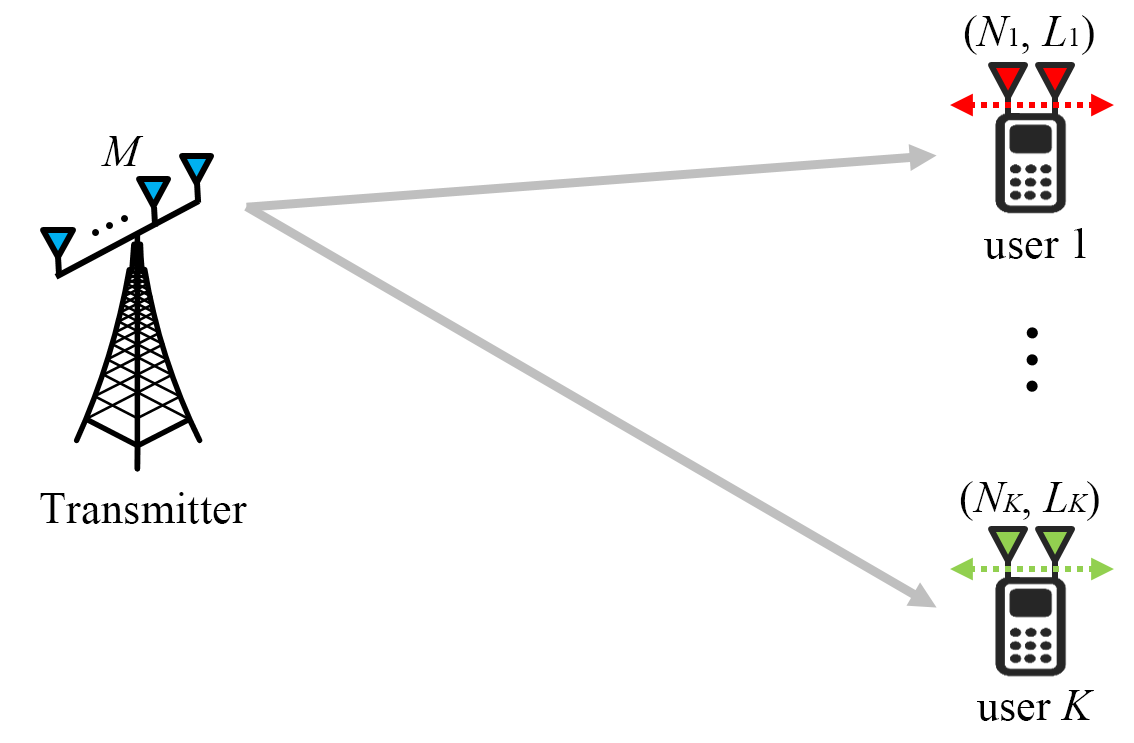}
\centering
\caption{$K$-user MIMO BC with reconfigurable antennas.}
\label{system}
\end{figure}

\subsection{Notation}
For integer values $a$ and $b$, $a\setminus b$ and $a|b$ denote the quotient and the remainder respectively when dividing $a$ by $b$.
For a set $\mathcal{A}$, $|\mathcal{A}|$ is the cardinality of $\mathcal{A}$. 
For a vector space $\mathcal{V}$, $\operatorname{dim}(\mathcal{V})$ is the dimension of $\mathcal{V}$.
For a matrix $\mathbf{A}$, $\mathbf{A}^{T}$, $|\mathbf{A}|$, $\operatorname{rank}(\mathbf{A})$, and  $\mathcal{R}(\mathbf{A})$ are the transpose, determinant, rank, and column space of $\mathbf{A}$ respectively. 
For matrices $\mathbf{A}$ and $\mathbf{B}$, $\mathbf{A}\otimes \mathbf{B}$ is the Kronecker product of $\mathbf{A}$ and $\mathbf{B}$.
For a set of matrices $\{\mathbf{A}_{i}\}_{i=1,\cdots,n}$, $\operatorname{diag}(\mathbf{A}_{1},\cdots,\mathbf{A}_{n})$ denotes the block-diagonal matrix consisting of $\{\mathbf{A}_{i}\}$.
Also $\mathbf{I}_{a}$, $\mathbf{1}_{a\times b}$, and $\mathbf{0}_{a\times b}$ denote the $a \times a$ identity matrix, the $a \times b$ all-one matrix, and the $a \times b$ all-zero matrix respectively and let $\mathbf{0}_{a} = \mathbf{0}_{a \times a}$.


\subsection{$K$-user MIMO BC with Reconfigurable Antennas}
Consider the $K$-user MIMO BC depicted in Fig. \ref{system} in which the transmitter is equipped with $M$ antennas and user $k \in \mathcal{K} = \{1,\cdots,K \}$ is equipped with a set of reconfigurable antennas whose RF-chains are able to choose $L_{k}$ receiving modes from $N_{k}$ preset modes at every time instant, where $N_{k} \geq L_{k}$.
Note that, if $N_{k} = L_{k}$, then user $k$ is equivalent to be equipped with $L_{k}$ conventional (non-reconfigurable) antennas.

The received signal vector of user $k$ at time $t$ is given by
\begin{align} \label{eq:in_out}
\mathbf{y}_{k}(t) = \boldsymbol{\Gamma}_{k}(t) \mathbf{H}_{k}(t) \mathbf{x}(t) + \mathbf{z}_{k}(t) 
\end{align}
where $\mathbf{H}_{k}(t) \in \mathbb{C}^{N_{k}\times M}$ is the channel matrix from the transmitter to $N_k$ preset modes of user $k$ at time $t$, $\mathbf{x}(t) \in \mathbb{C}^{M} $ is the transmit signal vector at time $t$, $\mathbf{z}_{k}(t) \in \mathbb{C}^{L_{k}}$ is the additive noise vector of user $k$ at time $t$, and $\boldsymbol{\Gamma}_{k}(t) \in \{0,1\}^{L_{k} \times N_{k}}$ is the selection matrix of user $k$ at time $t$.
In particular, each row vector of $\boldsymbol{\Gamma}_{k}(t)$ consists of zero values except for a single element of one value and is different from each other.
That is, $\mathbf{\Gamma}_k(t)$ extracts $L_{k}$ elements out of the $N_{k}$ elements in $\mathbf{H}_{k}(t)\mathbf{x}(t)$ 
and if user $k$ is equipped with conventional antennas, i.e., $L_{k} = N_{k}$, then $\boldsymbol{\Gamma}_{k}(t) = \mathbf{I}_{N_{k}}$ so that $\boldsymbol{\Gamma}_{k}(t)$ can be omitted in \eqref{eq:in_out}.
The transmitter should satisfy the average power constraint $P$, i.e., $\lim_{n \rightarrow \infty} \frac{1}{n}\sum_{t=1}^{n} \|\mathbf{x}(t)\|^{2}\leq P$, where $\|\cdot\|$ denote the norm of a vector. The elements of $\mathbf{z}_{k}(t)$ are independent and identically distributed  (i.i.d.) drawn from $\mathcal{CN}(0,1)$.

We assume that channel coefficients are i.i.d. drawn from a continuous distribution and remain constant across time, i.e., $\mathbf{H}_{k}(t) = \mathbf{H}_{k}$ for all $t \in \mathbb{N}$.
Global channel state information (CSI) is assumed to be available only at the users, but not at the transmitter. i.e., no CSIT.
Furthermore, we assume that each user selects its receiving modes in a predetermined pattern independent of channel realization, which are revealed to the transmitter.
That is, $\boldsymbol{\Gamma}_{k}(t)$ is not a function of $\{\mathbf{H}_j\}_{j\in \mathcal{K}}$ for all $k\in\mathcal{K}$ and $t\in\mathbb{N}$.

For notational convenience, from \eqref{eq:in_out}, we define the $n$ time-extended input--output relation as
\begin{align} \label{eq:in_out_extended}
\mathbf{y}^n_{k} = \boldsymbol{\Gamma}^n_{k} \mathbf{H}^n_{k} \mathbf{x}^n + \mathbf{z}^n_{k} 
\end{align}
where
\begin{align*}
\mathbf{\Gamma}^n_k&=\operatorname{diag}\left(\mathbf{\Gamma}_k(1),\cdots,\mathbf{\Gamma}_k(n)\right),\\
\mathbf{H}^n_k&= \mathbf{I}_{n} \otimes \mathbf{H}_k, \\ 
\mathbf{y}^n_{k}&=\left[\mathbf{y}^T_{k}(1)\cdots\mathbf{y}^T_{k}(n)\right]^T,\\
\mathbf{x}^n_{k}&=\left[\mathbf{x}^T_{k}(1)\cdots\mathbf{x}^T_{k}(n)\right]^T,\\
\mathbf{z}^n_{k}&=\left[\mathbf{z}^T_{k}(1)\cdots\mathbf{z}^T_{k}(n)\right]^T.
\end{align*}

\section{Linear Degrees of Freedom and Main Results}\label{sec:main_results}

\subsection{Linear Degrees of Freedom}
In this paper, we confine the transmitter to use linear precoding techniques, 
in which DoF represents the dimension of the linear subspace of transmitted signals \cite{Lashgari:14}.
Consider a linear precoding scheme with block length $n$, in which the transmitter sends the information symbols  of user $k$, denoted by $\mathbf{s}_{k} \in \mathbb{C}^{m_{k}(n)}$, through the $n$ time-extended beamforming matrix $\mathbf{V}^n_{k} \in \mathbb{C}^{n M \times m_{k}(n)}$.
Hence, the $n$ time-extended transmit signal vector is given by
\[\mathbf{x}^{n} = \sum_{j=1}^K\mathbf{V}_{j}^{n}\mathbf{s}_{j}\]
and, from \eqref{eq:in_out_extended}, the $n$ time-extended received signal vector of user $k$ is given by
\begin{align}
\mathbf{y}_{k}^{n} = \sum\limits_{j=1}^{K} \boldsymbol{\Gamma}_{k}^{n}\mathbf{H}_{k}^{n}\mathbf{V}_{j}^{n}\mathbf{s}_{j} + \mathbf{z}_{k}^{n}. \notag
\end{align}

Based on such a linear precoding scheme, we define the linear degrees of freedom as the follow,  see also \cite{Lashgari:14} for more details.
\begin{definition} \label{def:LDoF}
The linear degrees of freedom (LDoF) of $K$-tuple ($d_{1}, \cdots, d_{K}$) is said to be achievable if there exist a set of beamforming matrices $\mathbf{V}_{j}^{n}$ and selection matrices $\boldsymbol{\Gamma}_{j}^{n}$ for $j=1,\cdots,K$ almost surely satisfying 
\begin{gather}
\dim\left ( \operatorname{Proj}_{\mathcal{I}_{j}^{c}}\mathcal{R}(\boldsymbol{\Gamma}_{j}^{n}\mathbf{H}_{j}^{n}\mathbf{V}_{j}^{n}) \right ) = m_{j}(n), \notag \\
d_{j} = \underset{n \rightarrow \infty}{\lim} \frac{m_{j}(n)}{n} \notag
\end{gather}
where $\mathcal{I}_{j} = \mathcal{R}( \boldsymbol{\Gamma}_{j}^{n}\mathbf{H}_{j}^{n}[\mathbf{V}_{1}^{n}\cdots\mathbf{V}_{j-1}^{n}\mathbf{V}_{j+1}^{n}\cdots\mathbf{V}_{K}^{n}] )$ and
$\operatorname{Proj}_{\mathcal{A}^{c}}\mathcal{B}$ denotes the vector space induced by projecting the vector space $\mathcal{B}$ onto the orthogonal complement of the vector space $\mathcal{A}$. 
\end{definition}

The LDoF region $\mathcal{D}$ is the closure of the set of all achievable LDoF tuples satisfying  Definition \ref{def:LDoF} and  the sum LDoF is then given by
\begin{align}
d_{\Sigma} = \max_{(d_{1},\cdots,d_{K}) \in \mathcal{D}} \left\{\sum_{k=1}^{K}d_{k} \right\} \notag .
\end{align}

\subsection{Main Results}
For convenience of representation, the following parameters are defined.
\begin{align}
L_{\max} &= \max_{k\in\mathcal{K}}\{L_{k}\}, \notag \\
T_{k} & = \min(M,N_{k}) \mbox{ for } k \in \mathcal{K}, \notag  \\
\Lambda &= \{ k \in \mathcal{K}  :  T_{k} > L_{\max} \}, \notag \\
\eta & =\frac{\sum\limits_{i \in \Lambda}\frac{T_{i}L_{i}}{T_{i}-L_{i}}}{1+\sum\limits_{i\in\Lambda}\frac{L_{i}}{T_{i}-L_{i}}}. \label{Lambda}
\end{align}

In the following, we completely characterize the sum LDoF of the $K$-user MIMO BC with reconfigurable antennas. 

\begin{theorem} 
\label{main_thm} 
For the K-user MIMO BC with reconfigurable antennas defined in Section \ref{sec:system_model}, the sum LDoF is given by
\begin{align}
d_{\Sigma} = \min ( M,\max(L_{\max}, \eta )  ). \label{eq:sumLDoF}
\end{align}
\end{theorem}
\begin{proof}
We refer to Section \ref{subsec:converse1} for the converse proof and Section \ref{achievability_thm1} for the achievability proof.  
\end{proof}

\begin{remark}
From Theorem \ref{main_thm}, $N_k$ greater than $M$ cannot further increase $d_{\Sigma}$. Therefore, the number of preset modes $N_k$ for maximizing $d_{\Sigma}$ is enough to set $N_k=M$ for $k\in\mathcal{K}$. Note that this remark is valid only in MIMO BC with reconfigurable antennas and it is shown in \cite{Wang:14} that the number of preset modes greater than that of transmit antennas can increase sum DoF in MIMO IC with reconfigurable antennas.
\end{remark}

\begin{figure}[!t]
\includegraphics[width=3.5in]{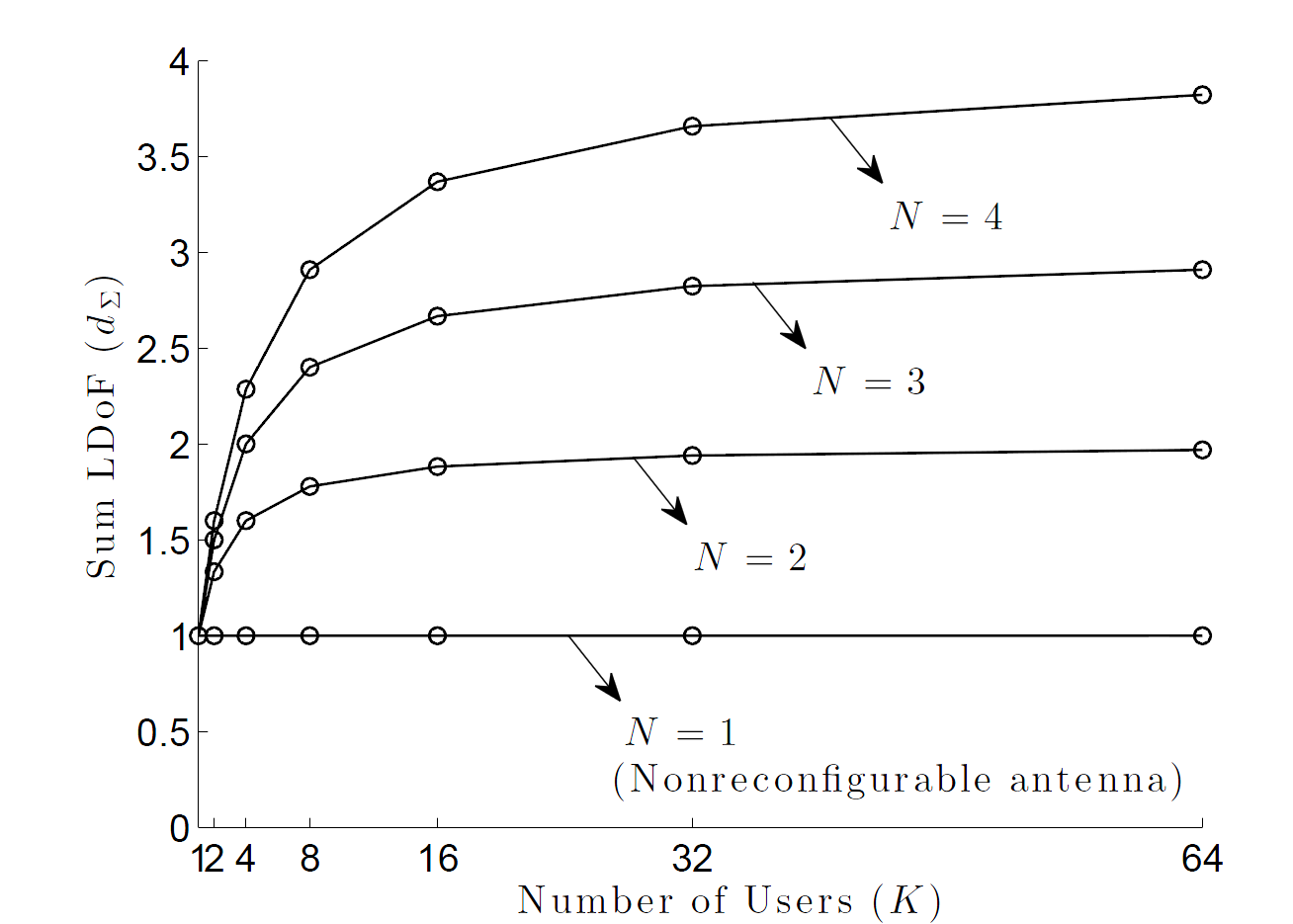}
\centering
\caption{Sum LDoF $d_{\Sigma}$ with respect to $K$ when $M=4$ and $L=1$.}
\label{ex1}
\end{figure}

\begin{example}Consider the symmetric $K$-user MIMO BC with reconfigurable antennas in Section \ref{sec:system_model} in which $N_k=N$ and $L_k=L$ for all $k\in\mathcal{K}$. 
For this case, 
\begin{align}
d_{\Sigma}=\min\left(M,\max\left(L,\frac{KL\min(M,N)}{KL+\min(M,N)-L}\right)\right)
\end{align}
from Theorem \ref{main_thm}.
To figure out the impact of reconfigurable antennas, let us focus on the limiting case where $K$ tends to infinity.
Then 
\begin{align}
\lim_{K\to\infty} d_{\Sigma}=\min(M,\max(L, \min(M,N)))=\min(M,N)
\end{align}
regardless of $L$.
Note that $d_{\Sigma}=\min(M,L)$ for the symmetric $K$-user MIMO BC without reconfigurable antennas, which corresponds to the case where $N=L$.
Therefore, reconfigurable antennas can significantly improve the sum LDoF as both $M$ and $N$ increase. Figure \ref{ex1} plots $d_{\Sigma}$ with respect to $K$ when $M=4$ and $L=1$. As the number of preset modes $N$ increases, the DoF gain from reconfigurable antennas increases compared to the conventional (nonreconfigurable) antenna model, i.e., $N=L$.  
\end{example}

We further derive the LDoF region $\mathcal{D}$ for a class of antenna configurations in the following theorem.
\begin{theorem}
\label{main_thm2}
Consider the $K$-user MIMO BC with reconfigurable antennas defined in Section \ref{sec:system_model}.
If $\ M > L_{\max}$ and $N_{k} > L_{\max}$ for all $ k \in \mathcal{K}$, then the LDoF region $\mathcal{D}$ consists of all $K$-tuples $(d_{1},\cdots, d_{K})$ satisfying
\begin{align}
\label{eq:ineq1}
\frac{d_{k}}{L_{k}}  + \sum\limits_{j =1, j\neq  k}^{K} \frac{d_{j}}{T_{j}}  \leq 1
\end{align}
for all $k \in \mathcal{K}$.
\end{theorem}
\begin{proof}
We refer to Section \ref{subsec:converse2} for the converse proof and Section \ref{achievability_thm2} for the achievability proof.  
\end{proof}

\begin{figure}[!t]
\includegraphics[width=3in]{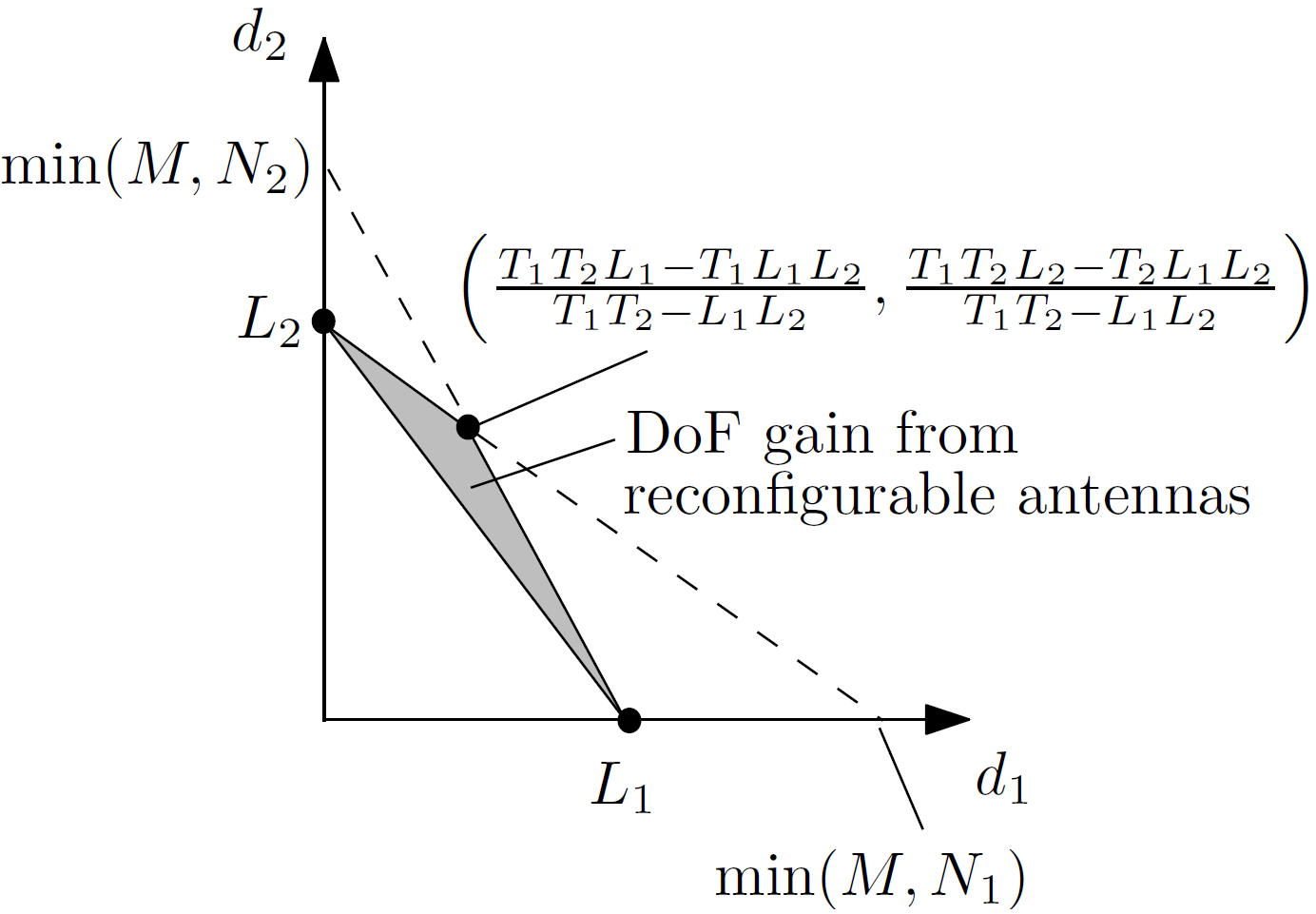}
\centering
\caption{LDoF region $\mathcal{D}$ for the $2$-user MIMO BC with reconfigurable antennas, where $M,N_1,N_2>\max(L_1,L_2)$.}
\label{dof_region}
\end{figure}

\begin{figure}[!t]
\includegraphics[width=4in]{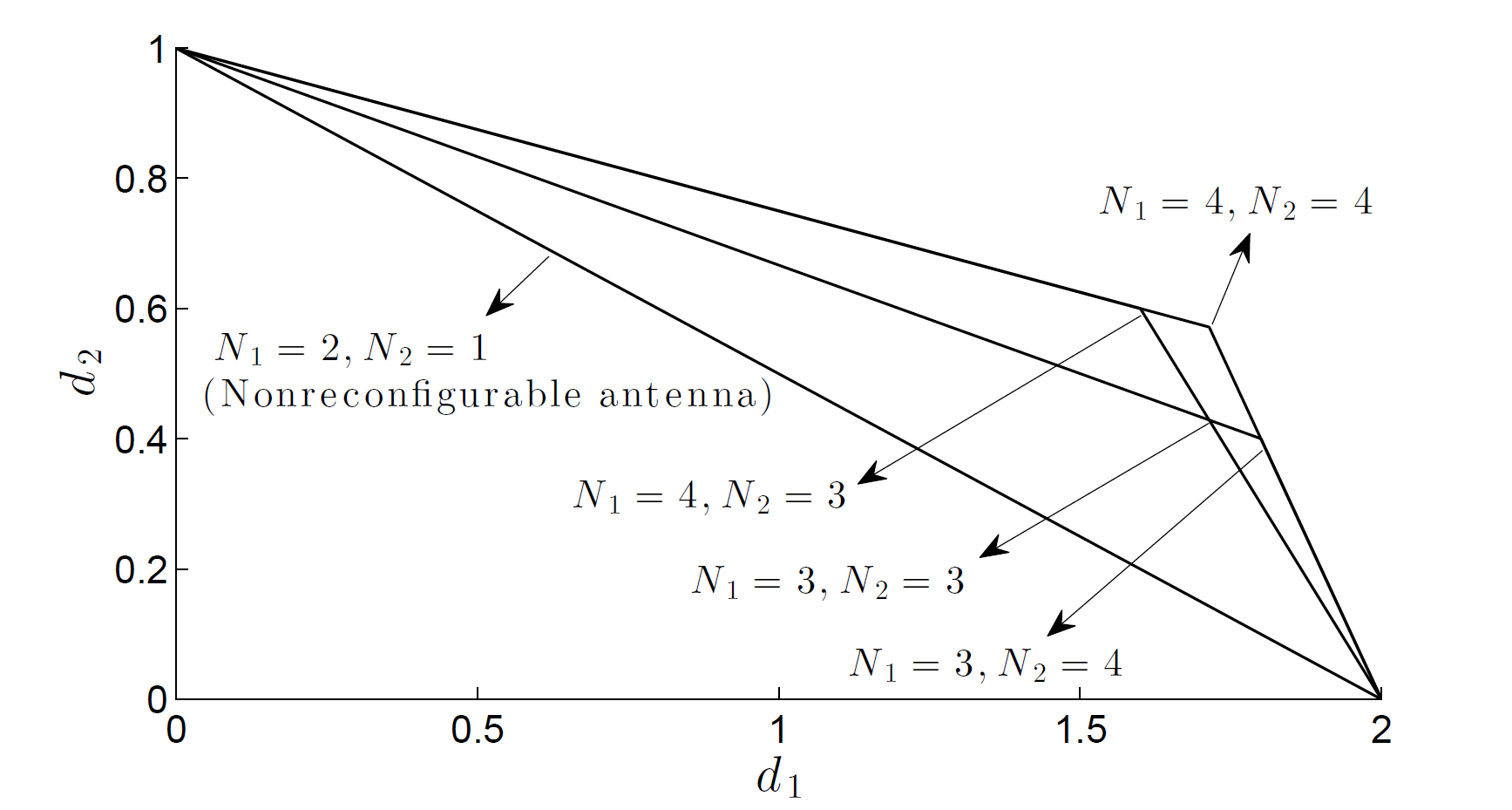}
\centering
\caption{LDoF region $\mathcal{D}$ when $K=2$, $M=4$, $L_1=2$, and $L_2=1$.}
\label{ex2}
\end{figure}

\begin{example}
Consider the $2$-user MIMO BC with reconfigurable antennas in Section \ref{sec:system_model} in which $M,N_1,N_2>\max(L_1,L_2)$.
From Theorem \ref{main_thm2}, the LDoF region $\mathcal{D}$ is then given as in Fig. \ref{dof_region}. For the conventional (nonreconfigurable) antenna model, where $N_1=L_1$ and $N_2=L_2$, $\mathcal{D}$ is given by the time-sharing region between $(L_1,0)$ and $(0,L_2)$.
Hence $\mathcal{D}$ enlarges as $N_1$ and $N_2$ increase, which demonstrate the benefit of reconfigurable antennas. Figure \ref{ex2} plots $\mathcal{D}$ when $K=2$, $M=4$, $L_1=2$, and $L_2=1$.
\end{example}

From Theorem \ref{main_thm}, the sum LDoF is derived for a class of the $K$-user MIMO IC with reconfigurable antennas in the following.
We omit the formal definition of LDoF for the $K$-user MIMO IC with reconfigurable antennas, which can be straightforwardly defined in the same manner as in Definition \ref{def:LDoF}.

\begin{figure}[!t]
\includegraphics[width=3in]{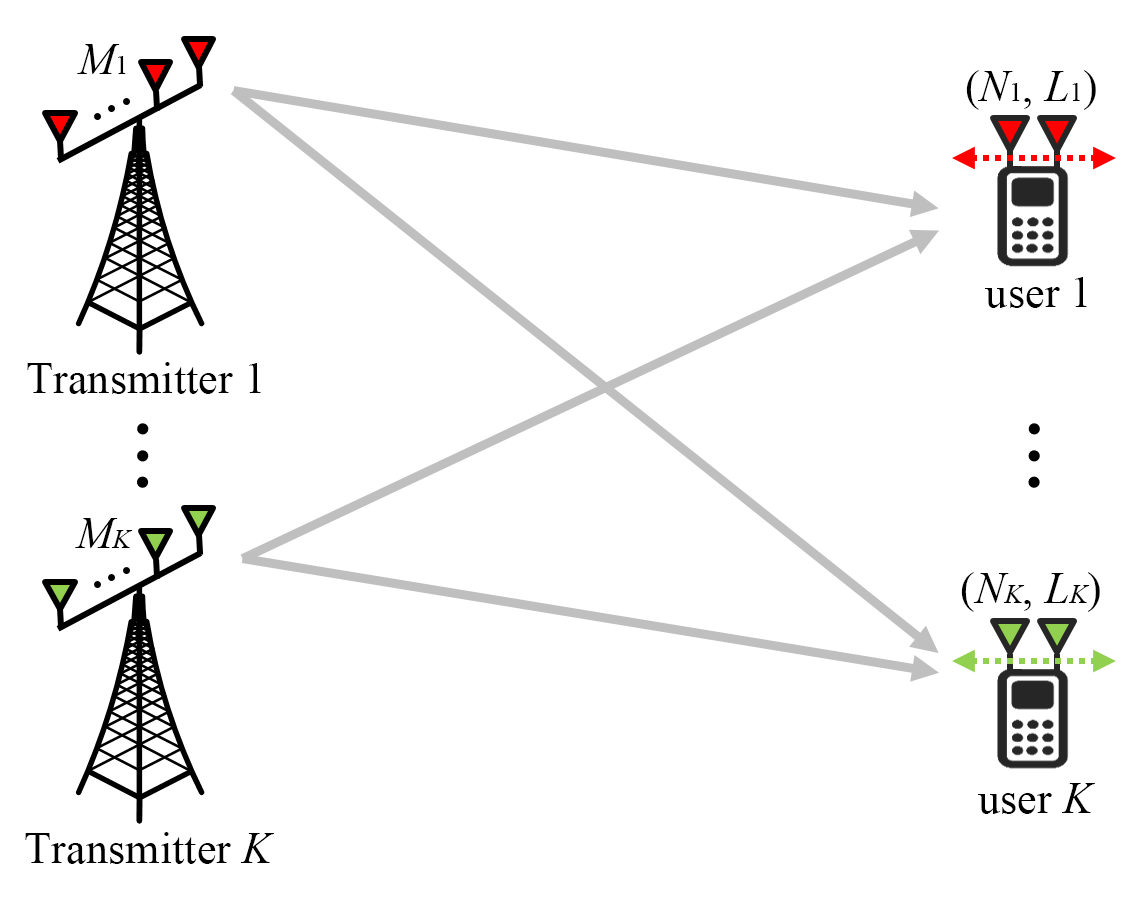}
\centering
\caption{$K$-user MIMO IC with reconfigurable antennas.}
\label{system_IC}
\end{figure}


\begin{corollary} \label{main_thm3}
Consider the $K$-user MIMO IC with reconfigurable antennas depicted in Fig. \ref{system_IC} in which transmitter $k \in \mathcal{K}$ is equipped with $M_{k}$ antennas and user $k$ is equipped with a set of reconfigurable antennas whose RF-chains are able to choose $L_{k}$ receiving mode from $N_{k}$ preset modes, where $N_{k} \geq L_{k}$. 
If $M_{k} \geq N_{k}$ for all $k \in \mathcal{K}$, then the sum LDoF  is given by 
\begin{align} \label{eq:dof_ic} 
d_{\Sigma,\text {IC}} = \max(L_{\max}, \eta_{\text {IC}})
\end{align}
where $\eta_{\text {IC}}$ is defined as $\eta$ with $\Lambda = \{ k \in \mathcal{K} : N_{k} > L_{\max} \}$ and $T_{k} = N_{k}$ for all $k \in \Lambda$
\end{corollary}
\begin{IEEEproof}
Obviously, the achievable LDoF of the $K$-user MIMO IC with reconfigurable antennas defined in Corollary \ref{main_thm3} is upper bounded by $d_{\Sigma}$ of the $K$-user MIMO BC with reconfigurable antennas where the transmitter is equipped with $\sum_{k=1}^K M_k$ antennas and user $k \in \mathcal{K}$ is equipped with a set of reconfigurable antennas whose RF-chains are able to choose $L_{k}$ receiving modes from $N_{k}$ preset modes.
Hence, from Theorem \ref{main_thm}, the LDoF of the considered $K$-user MIMO IC is upper bounded by \eqref{eq:dof_ic}, which completes the converse proof of Corollary \ref{main_thm3}.
We refer to Section \ref{achievability_thm3} for the achievability proof. 
\end{IEEEproof}

\begin{example} Consider the symmetric MIMO IC with reconfigurable antennas in Fig. \ref{system_IC} in which $N_k=N$ and $L_k=L$ for all $k\in\mathcal{K}$ ,where $N\geq L$. If $M\geq N$, then from Corollary \ref{main_thm3}, 
\begin{align}
d_{\Sigma, \text{IC}}=\max\left(L,\frac{KLN}{KL+N-L}\right),
\end{align}
which attains $\lim_{K\to \infty}d_{\Sigma, \text{IC}}=N$.
Note that the symmetric $K$-user MIMO IC without reconfigurable antennas is given by $d_{\Sigma, \text{IC}}=L$, which corresponds to the case where $M\geq N=L$.
Therefore, similar to the symmetric MIMO BC case, reconfigurable antenna can significantly improve the sum LDoF as both $M$ and $N$ increase with $M\geq N$.
\end{example}

The following two remarks summarize the contributions of Theorem \ref{main_thm} and Corollary \ref{main_thm3}, compared with the previous results in \cite{Wang:10, Lu:13}.

\begin{remark}
Consider the $K$-user MIMO BC with reconfigurable antennas defined in Section \ref{sec:system_model}.
If $M = N_{k}$ and $L_{k} = L$ for all $k \in \mathcal{K}$ where $M > L$, then 
\begin{align}
d_{\Sigma}=\frac{MLK}{M+LK-L}\notag
\end{align}
from Theorem \ref{main_thm}, which coincides with the previous achievability result in \cite{Wang:10}.
Hence, Theorem \ref{main_thm} not only generalizes the result in \cite{Wang:10} but it also shows the converse in the LDoF sense for general $M$, $\{N_k\}_{k\in\mathcal{K}}$, and $\{L_k\}_{k\in\mathcal{K}}$.
\end{remark}

\begin{remark}
Consider the $K$-user MIMO IC with reconfigurable antennas defined in Corollary \ref{main_thm3}.
If  $M_k=N_k > 1$ and $L_k = 1$ for all $k\in\mathcal{K}$, then
\begin{align}
d_{\Sigma,\text{IC}} =\frac{\sum\limits_{k= 1}^K\frac{N_{k}}{N_{k}-1}}{1+\sum\limits_{k=1}^K\frac{1}{N_{k}-1}} \notag
\end{align} 
from Corollary \ref{main_thm3}, which coincides with the previous achievability result in \cite{Lu:13}.
Hence, Corollary \ref{main_thm3} not only  generalizes the result in \cite{Lu:13} but it also shows the converse in the LDoF sense for a broader class of antenna configurations.
\end{remark}


\section{Converse} \label{sec:converse}
In this section, we prove the converse of Theorem  \ref{main_thm}, \ref{main_thm2}.

\subsection{Converse of Theorem \ref{main_thm}} \label{subsec:converse1}

First divide the entire parameter space into three cases as follows:
\begin{itemize}
\item Case 1: $M \leq L_{\max}$.
\item Case 2: $M > L_{\max}$ and $N_{k} \leq L_{\max}$ for all $ k \in \mathcal{K}$.
\item Case 3: $M > L_{\max}$ and $N_{k} > L_{\max}$ for some $k \in \mathcal{K}$.
\end{itemize}
Then the right hand side of \eqref{eq:sumLDoF} is given by
\begin{align}
\min ( M,\max(L_{\max}, \eta )  )=
\begin{cases}
M & \mbox{for Case 1}, \\
L_{\max} & \mbox{for Case 2}, \\
\max ( L_{\max} , \eta )  & \mbox{for Case 3}. \\
\end{cases}
\label{eq:converse_total}
\end{align}

For Case 1, an achievable sum LDoF is trivially upper bounded by the number of transmit antennas. 
Consequently, we have
\begin{align}
d_{\Sigma} \leq M & \ \mbox{ for Case 1}. \label{eq:converse_case1}
\end{align}

For Case 2, consider the \emph{extended $K$-user MIMO BC} by substituting $N_{k} = L_{\max}$ and $L_{k}=L_{\max}$ for all $k \in \mathcal{K}$ from the original $K$-user MIMO BC with reconfigurable antennas. 
That is, for the extended $K$-user MIMO BC, all users are equipped with $L_{\max}$ conventional antennas. 
Obviously, the sum DoF of the extended $K$-user MIMO BC provides an upper bound on $d_{\Sigma}$.
From the fact that the received signals of all the user are statistically equivalent in the extended $K$-user MIMO BC so that any receiver can decode all messages from the transmitter, $d_{\Sigma}$ is further bounded by the sum DoF of point-to-point MIMO BC where transmitter and receiver are equipped with $M$ and $L_{\max}$ conventional antennas respectively, given by $\min(M, L_{\max}) = L_{\max}$. Therefore, we have
\begin{align}
d_{\Sigma} \leq L_{\max} & \ \mbox{ for Case 2}. \label{eq:converse_case2}
\end{align}

Hence, for the rest of this subsection, we prove that 
\begin{align}
d_{\Sigma} & \leq \max (L_{\max},  \eta) \notag
\end{align}
by assuming that $M > L_{\max}$ and $N_{k} > L_{\max}$ for some $k \in \mathcal{K}$, which is Case 3.
Suppose that user $k$ satisfies the condition $T_k>L_{\max}$ (equivalently $k\in\Lambda$).
Then, consider the \emph{extended $K$-user MIMO BC with reconfigurable antennas at user $k$} by substituting 
$N_{i}=L_{\max}$ and $L_{i}=L_{\max}$ for all $i \in \mathcal{K} \setminus \Lambda$ and $L_{i} = N_{i}$ for all $i \in \Lambda\setminus\{k\}$ from the original $K$-user MIMO BC with reconfigurable antennas.
Hence, users in $\mathcal{K} \setminus \Lambda$ have $L_{\max}$ conventional antennas and user $i \in \Lambda\setminus\{k\}$ has $N_{i}$ conventional antennas.
Only user $k$ is equipped with reconfigurable antennas in this extended model.
Again, the sum DoF of this model provides an upper bound on $d_{\Sigma}$. 
Then, the received signal vector of user $i$ is given by
\begin{align} 
\mathbf{y}_{i}(t) =\begin{cases}
\boldsymbol{\Gamma}_{k}(t)\mathbf{H}_{k} \mathbf{x}(t) + \mathbf{z}_{k}(t) & \mbox{if } i = k, \\
\mathbf{G}_{i}\mathbf{x}(t) + \mathbf{z}_{i}(t) & \textrm{otherwise}
\end{cases}
\end{align}
where $\mathbf{G}_{i} \in \mathbb{C}^{\max(N_{i},L_{\max}) \times M}$ for $i \in \mathcal{K}\setminus \{k\}$ satisfies the channel assumption in Section \ref{sec:system_model}.

\begin{figure}[!t]
\includegraphics[width=3in]{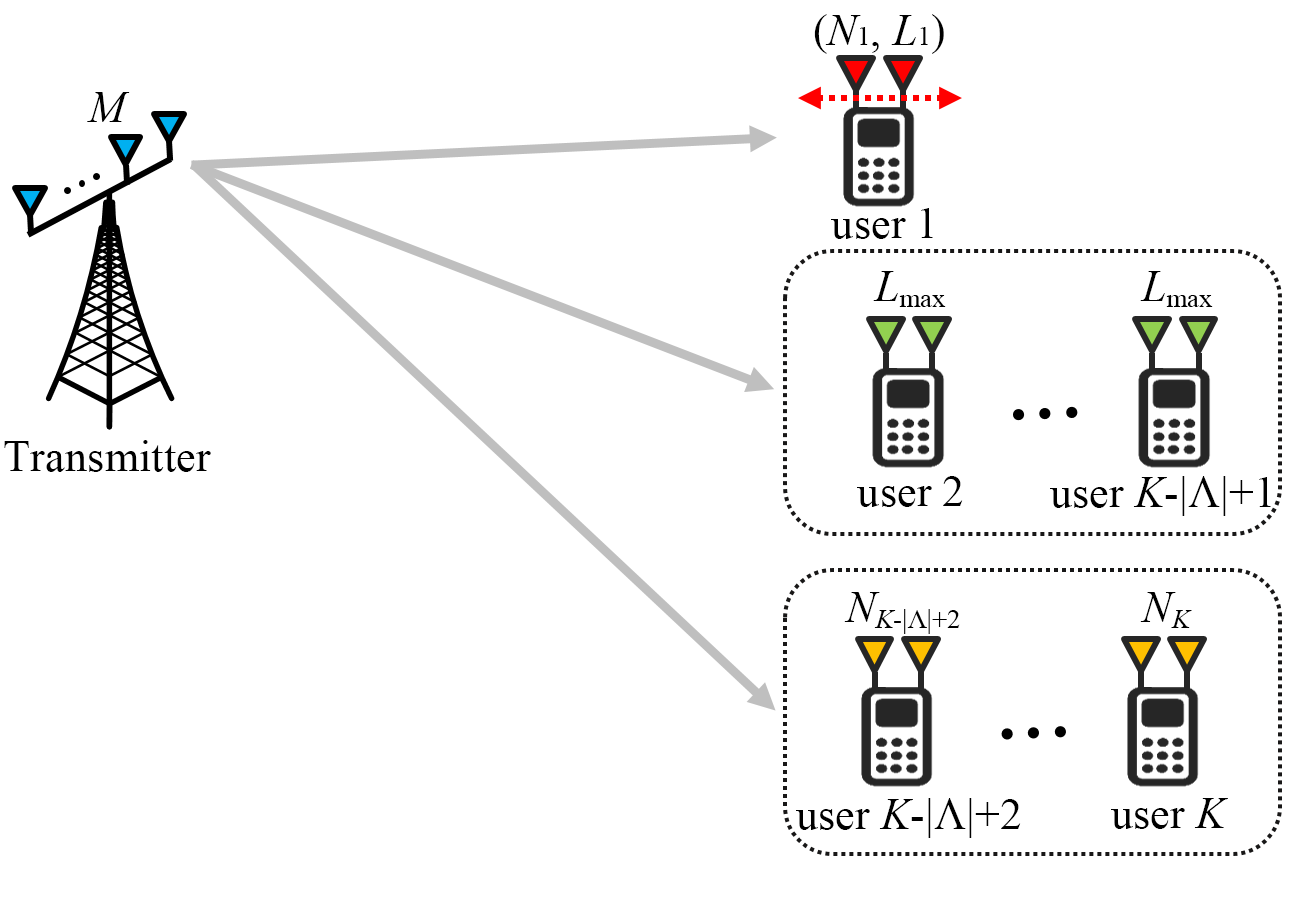}
\centering
\caption{Extended $K$-user MIMO BC with reconfigurable antennas based on the rearranged user index, where we assume that if $|\Lambda|=1$, then the set of user $K-|\Lambda|+2$ through user $K$ is empty and if $|\Lambda|=K$, then the set of user 2 through $K-|\Lambda|+1$ is empty.}
\label{system_extended}
\end{figure}

For convenience, we rearrange the users in ascending order of $L_{i}$ and denote the new index of user $i$ as $\sigma(i)$.
We assume that $\sigma(k)=1$ without loss of generality.
From now on, we denote the index $i$ as the rearranged user index. Fig. \ref{system_extended} illustrates the extended model based on the rearranged user index.
Hence, the $n$ time-extended received signal vector of user $i$ with linear precoding is given by 
\begin{align} \label{eq:ch_extended}
\mathbf{y}_{i}^{n} = \begin{cases}
\sum\limits_{j=1}^{K} \boldsymbol{\Gamma}_{i}^{n}\mathbf{H}_{i}^{n}\mathbf{V}_{j}^{n}\mathbf{s}_{j} + \mathbf{z}_{i}^{n} &\mbox{if } i = 1, \\
\sum\limits_{j=1}^{K} \mathbf{G}_{i}^{n}\mathbf{V}_{j}^{n}\mathbf{s}_{j} + \mathbf{z}_{i}^{n} & \textrm{otherwise} \\
\end{cases}
\end{align}
where $\mathbf{G}_{i}^{n}= \mathbf{I}_{n} \otimes \mathbf{G}_{i}$ for $i \in \mathcal{K}\setminus \{1\}$.
Also, we define an increasing sequence $\Delta_{i}$ for $i\in\mathcal{K}$ as
\begin{align}
\Delta_{i} = \begin{cases}
L_{k} & \mbox{if }   \sigma^{-1}(i) = k,\\
L_{\max} & \mbox{if } \sigma^{-1}(i) \in \mathcal{K}\setminus \Lambda, \\
T_{\sigma^{-1}(i)} & \textrm{otherwise}.
\end{cases} \notag
\end{align}
Note that $\Delta_{i}$ is the rank of $\mathbf{\Gamma}_1(t)\mathbf{H}_1$ for $i=1$ and the rank of $\mathbf{G}_{i}$ for $i = 2,\cdots, K$, almost surely.

In the following, we introduce three key lemmas used for proving the converse of Theorem \ref{main_thm}.
The first lemma provides an equivalent condition for decodability of messages \cite{Lashgari:14}.
\begin{lemma}[Lashgari--Avestimehr--Suh]
\label{lemma_fano}
For two matrices $\mathbf{A}$, $\mathbf{B}$ with the same row size,
\begin{align}
\operatorname{dim}(\operatorname{Proj}_{\mathcal{R}(\mathbf{A})^{c}}\mathcal{R}(\mathbf{B})) = \operatorname{rank} \left ([ \mathbf{A} \  \mathbf{B} \right ]) - \operatorname{rank} \left(\mathbf{A} \right). \notag
\end{align}
\end{lemma}
\begin{proof}
We refer to \cite[Lemma 1]{Lashgari:14} for the proof.
\end{proof}

The second lemma states that mode switching does not decrease the dimension of the interference space of user 1 almost surely.
\begin{lemma}
\label{lemma_sel}
Consider the extended $K$-user MIMO BC with reconfigurable antennas at user 1 depicted in Fig. \ref{system_extended}. Let $\mathbf{G}_{1}\in\mathbb{C}^{L_1\times M}$ denote the matrix consisting of the first through the $L_{1}$th row vectors of $\mathbf{H}_{1}$ and $\mathbf{G}_{1}^{n}= \mathbf{I}_{n} \otimes \mathbf{G}_1$. 
For any mode switching pattern $\boldsymbol{\Gamma}_{1}^{n}$, the following relation holds almost surely:
\begin{align}
\operatorname{rank} \left (\boldsymbol{\Gamma}_{1}^{n}\mathbf{H}_{1}^{n}[\mathbf{V}_{2}^{n} \cdots \mathbf{V}_{K}^{n}] \right ) \geq  \operatorname{rank} \left (\mathbf{G}_{1}^{n}[\mathbf{V}_{2}^{n} \cdots \mathbf{V}_{K}^{n}] \right ). \label{eq_sel}
\end{align}
\end{lemma}
\begin{IEEEproof}
We refer to Appendix \ref{proof_sel} for the proof.
\end{IEEEproof}
Although the definition of $\mathbf{G}^n_{1}$ in Lemma \ref{lemma_sel} is not consistent with those of $\mathbf{G}^n_{i}$ for $i = 2,\cdots, K$ in \eqref{eq:ch_extended}, we adopt this notation for easy presentation of the converse proof.
The third lemma shows the relation of the dimensions of the interference space between user $i$ and user $i-1$.
\begin{lemma}
\label{lemma_equi}
Consider the extended $K$-user MIMO BC with reconfigurable antennas at user 1 depicted in Fig. \ref{system_extended}. The following relations hold almost surely:
\begin{align} \label{eq: relation_1}
\frac{1}{\Delta_{i-1}}\operatorname{rank} \left (\mathbf{G}_{i-1}^{n}[\mathbf{V}_{i}^{n} \cdots \mathbf{V}_{K}^{n}] \right ) 
&\geq \frac{1}{\Delta_{i}}\operatorname{rank} \left (\mathbf{G}_{i}^{n}[\mathbf{V}_{i+1}^{n} \cdots \mathbf{V}_{K}^{n}] \right ) 
+ \frac{1}{\Delta_{i}} \operatorname{dim}(\operatorname{Proj}_{\mathcal{I}_{i}^{c}}\mathcal{R}(\mathbf{G}_{i}^{n}\mathbf{V}_{i}^{n})),
\end{align}
for  $i = 2,\cdots, K-1$ and 
\begin{align} \label{eq: relation_2}
\frac{1}{\Delta_{K-1}}\operatorname{rank} \left (\mathbf{G}_{K-1}^{n}\mathbf{V}_{K}^{n} \right ) 
&\geq \frac{1}{\Delta_{K}} \operatorname{dim}(\operatorname{Proj}_{\mathcal{I}_{K}^{c}}\mathcal{R}(\mathbf{G}_{K}^{n}\mathbf{V}_{K}^{n})). 
\end{align}
\end{lemma}
\begin{IEEEproof}
We refer to Appendix \ref{proof_equi} for the proof.
\end{IEEEproof}

We are now ready to prove the converse of Theorem \ref{main_thm}. From the definition of $m_{1}(n) $, we have
\begin{subequations}
\begin{align}
m_{1}(n) 
&=  \operatorname{dim}(\operatorname{Proj}_{\mathcal{I}_{1}^{c}}\mathcal{R}(\boldsymbol{\Gamma}_{1}^{n}\mathbf{H}_{1}^{n}\mathbf{V}_{1}^{n})) \notag \\
&=  \operatorname{rank} \left (\boldsymbol{\Gamma}_{1}^{n}\mathbf{H}_{1}^{n}[\mathbf{V}_{1}^{n} \cdots \mathbf{V}_{K}^{n}] \right ) - \operatorname{rank}\left (\boldsymbol{\Gamma}_{1}^{n}\mathbf{H}_{1}^{n}[\mathbf{V}_{2}^{n} \cdots \mathbf{V}_{K}^{n}] \right ) \label{eq_dof11}\\
&\leq  n \Delta_{1} - \operatorname{rank}\left (\boldsymbol{\Gamma}_{1}^{n}\mathbf{H}_{1}^{n}[\mathbf{V}_{2}^{n} \cdots \mathbf{V}_{K}^{n}] \right ) \notag \\
&\overset{a.s.}{\leq}  n \Delta_{1}  - \operatorname{rank} \left (\mathbf{G}_{1}^{n}[\mathbf{V}_{2}^{n} \cdots \mathbf{V}_{K}^{n}] \right ) \label{eq_dof12}\\
&\overset{a.s.}{\leq} n \Delta_{1} - \sum\limits_{i=2}^{K} \frac{\Delta_{1}}{\Delta_{i}} \operatorname{dim}(\operatorname{Proj}_{\mathcal{I}_{i}^{c}}\mathcal{R}(\mathbf{G}_{i}^{n}\mathbf{V}_{i}^{n})) \label{eq_dof14} \\
&= n \Delta_{1} - \sum\limits_{i=2}^{K} \frac{\Delta_{1}}{\Delta_{i}} m_{i}(n) \label{eq_dof15}
\end{align}
\end{subequations}
where \eqref{eq_dof11}, \eqref{eq_dof12}, \eqref{eq_dof14}, and \eqref{eq_dof15} follow from Lemma \ref{lemma_fano}, \ref{lemma_sel}, \ref{lemma_equi}, and Definition \ref{def:LDoF} respectively.
Then, by dividing both sides by $n$ and letting $n$ to infinity, we have
\begin{align}
\label{dof_result21}
\sum\limits_{i=1}^{K} \frac{1}{\Delta_{i}} d_{i} \leq 1.
\end{align}

Rearranging \eqref{dof_result21} with respect to the original index, i.e., $\sigma^{-1}(i)$ provides
\begin{align}
\label{dof_result2}
\frac{1}{L_{k}} d_{k} + \sum\limits_{i \in \Lambda\setminus \{ k\} } \frac{1}{T_{i}} d_{i}+\sum\limits_{i \notin \Lambda } \frac{1}{L_{\max}} d_{i} \leq 1.
\end{align}

Since \eqref{dof_result2} holds for all $k \in \Lambda$, we have total $|\Lambda|$ inequalities composing the outer region of $\mathcal{D}$.
Then, we obtain an upper bound on $d_{\Sigma}$ by solving the linear programming in the following lemma.

\begin{lemma} \label{lemma:linear_program}
\label{lemma_upper} Consider the following optimization problem assuming that $M > L_{\max}$ and $N_{k} > L_{\max}$ for some $k \in \mathcal{K}$:
\begin{align}
\operatorname{maximize} \sum\limits_{i=1}^{K} d_{i} \notag 
\end{align}
\begin{align}
\textrm{subject to } & \frac{1}{L_{k}} d_{k} + \sum\limits_{i \in \Lambda\setminus \{ k\} } \frac{1}{T_{i}} d_{i}+\sum\limits_{i \notin \Lambda } \frac{1}{L_{\max}} d_{i} \leq 1, \ \ \forall k\in \Lambda,
 \notag \\
& d_i \geq 0, \ \ \forall i\in\mathcal{K}.  \notag
\end{align}
Then 
\begin{align}
\sum\limits_{i=1}^{K} d_{i} \leq 
\max \left ( \eta, L_{\max} \right ) \notag.
\end{align}
\end{lemma}
\begin{IEEEproof}
We refer to Appendix \ref{proof_upper} for the proof.
\end{IEEEproof}

Therefore, combining \eqref{eq:converse_total}, \eqref{eq:converse_case1}, \eqref{eq:converse_case2} and the result in Lemma \ref{lemma:linear_program} provides
\begin{align} 
d_{\Sigma} \leq \min( M , \max ( L_{\max}, \eta ) )\notag,
\end{align} 
which completes the converse proof of Theorem \ref{main_thm}.


\subsection{Converse of Theorem \ref{main_thm2}} \label{subsec:converse2}
Notice that the condition $M > L_{\max}$ and $N_{k} > L_{\max}$ in Theorem \ref{main_thm2} is a special class of Case 3 defined in Section \ref{subsec:converse1} satisfying that $\Lambda=\mathcal{K}$.
Therefore, in the same manner in Section \ref{subsec:converse1}, we have  \eqref{dof_result2} with $\Lambda=\mathcal{K}$.
Therefore,
\begin{align}
\frac{d_{k}}{L_{k}}  + \sum\limits_{i =1, i\neq  k}^{K} \frac{d_{i}}{T_{i}}  \leq 1, & \ \  \forall k \in\mathcal{K}, \notag
\end{align}
which completes the converse proof of Theorem \ref{main_thm2}.

\section{Achievability}\label{sec:achievability}

In this section, we prove achievability of Theorems \ref{main_thm}, \ref{main_thm2}, and Corollary \ref{main_thm3}.
The proposed blind IA scheme generalizes those in \cite{Gou:10, Gou:11, Wang:10}, but it cannot be straightforwardly obtained from \cite{Gou:10, Gou:11, Wang:10} due to general antenna configurations of $M$, $\{N_{k}\}_{k\in\mathcal{K}}$, and $\{L_{k}\}_{k\in\mathcal{K}}$ considered in this paper.
For better understanding, we also provide an example for the proposed blind IA scheme based on the two-user case in Appendix \ref{appendix:example}.

\subsection{Achievability of Theorem \ref{main_thm}}
\label{achievability_thm1}

First divide the entire parameter space into four cases as follows:
\begin{itemize}
\item Case 1: $M \leq L_{\max}$.
\item Case 2: $M > L_{\max}$ and $N_{k} \leq L_{\max}$ for all $ k \in \mathcal{K}$.
\item Case 3-1: $M > L_{\max}$, $N_{k} > L_{\max}$ for some $k \in \mathcal{K}$, and $\eta \leq L_{\max}$,
\item Case 3-2: $M > L_{\max}$, $N_{k} > L_{\max}$ for some $k \in \mathcal{K}$, and $\eta > L_{\max}$,
\end{itemize}
where Cases 1 and 2 are identical to those in Section \ref{subsec:converse1} and Case 3-1 and Case 3-2 are two partitions of Case 3 in Section \ref{subsec:converse1}.
The right side of \eqref{eq:sumLDoF} is then given by
\begin{align}  \label{eq:dof_total}
\min( M , \max ( L_{\max} , \eta ) ) =
\begin{cases}
M & \mbox{for Case 1}, \\
L_{\max} & \mbox{for Case 2 or Case 3-1}, \\
\eta & \mbox{for Case 3-2}. \\
\end{cases} 
\end{align}

For Cases 1, 2, and 3-1, the sum DoF is trivially achievable by only supporting the user having the maximum number of RF-chains.
Hence, 
\begin{align} \label{eq:dof1_2}
d_{\Sigma} = 
\begin{cases}
M & \mbox{ for Case 1}, \\
L_{\max} & \mbox{ for Case 2 or Case 3-1} \\
\end{cases}
\end{align}
is achievable.
For the rest of this subsection, we prove that 
\begin{align}
d_{\Sigma} = \eta = \frac{\sum\limits_{k \in \Lambda}\frac{T_{k}L_{k}}{T_{k}-L_{k}}}{1+\sum\limits_{k \in \Lambda}\frac{L_{k}}{T_{k}-L_{k}}} 
\end{align}
is achievable by assuming that $M > L_{\max}$, $N_{i} > L_{\max}$ for some $i \in \mathcal{K}$, and $\eta > L_{\max}$, which is Case 3-2.
For this case, only the users in $\Lambda$ are supported, i.e., $d_{i} = 0$ for all $i \notin \Lambda$.
Suppose that $\Lambda=\{1,2,\cdots,|\Lambda|\}$ without loss of generality.
For easy representation, let us define $S_{i}$, $U_{i}$, and $W_{i}$ for $i \in \Lambda$ and define $U$ and $W$ as
\begin{align}
S_{i} & = \begin{cases}
\frac{T_{i}}{L_{i}}-1 & \mbox{if } T_{i}|L_{i} = 0, \\
T_{i} \setminus L_{i}  & \textrm{otherwise},
\end{cases} \notag 
\\
U_{i} & = S_{i}\prod_{p \in \Lambda\setminus\{i\}} (T_{p} - L_{p}), \notag
\\
W_{i} & = \prod_{p \in \Lambda\setminus\{i\}} S_{p}, \notag
\\
U & = \prod_{p \in \Lambda} (T_{p} - L_{p}), \notag
\\
W & =\prod_{p\in\Lambda}S_{p} \label{def:ach_parameter} .
\end{align}


\subsubsection{Transmit beamforming design} \label{subsubsec:transmit_signal}
To construct transmit beamforming, we adopt a bottom-up approach as in the following steps.

{\bf Step 1} (Alignment block): 
As the first step, we construct alignment blocks, which will be used for building alignment units in the next step.
The basic concept of alignment block in this paper is similar to those in \cite{Gou:11, Wang:10}.
Define $\mathbf{I}_{M,T_{i}} = [\mathbf{I}_{T_{i}} \ \mathbf{0}_{T_{i},M-T_{i}}]^{T}$ and the $j$th information vectors of user $i$ as $\mathbf{s}_{j}^{[i]}\in \mathbb{C}^{L_{i} T_{i}}$, which consists of $L_{i} T_{i}$ independent information symbols, where $j = 1, \cdots, U_{i}W_{i}$. 
Then the $j$th alignment block of user $i$, denoted by $\mathbf{v}_{j}^{[i]} \in \mathbb{C}^{MT_{i}}$, is defined as
\begin{align}
\mathbf{v}_{j}^{[i]} =\left [(\mathbf{v}_{j,1}^{[i]})^{T} (\mathbf{v}_{j,2}^{[i]})^{T} \cdots (\mathbf{v}_{j,S_{i}+1}^{[i]})^{T}\right]^{T} \notag
\end{align}
where for $k=1,\cdots,S_{i}+1$,
\begin{align} \label{eq:align_block}
\mathbf{v}_{j,k}^{[i]} =\begin{cases}
(\boldsymbol{\Phi}\otimes \mathbf{I}_{M,T_{i}})\mathbf{s}_{j}^{[i]}
\in \mathbb{C}^{(T_{i}|L_{i}) M} &\mbox{if } T_{i}|L_{i} \neq 0 \mbox{ and } k=1,\\
(\mathbf{I}_{L_{i}} \otimes \mathbf{I}_{M,T_{i}}) \mathbf{s}_{j}^{[i]}
\in \mathbb{C}^{L_{i} M}&\mbox{otherwise}.
\end{cases}
\end{align}
Here $\boldsymbol{\Phi} \in \mathbb{C}^{T_{i}|L_{i} \times L_{i}}$ is a random matrix whose entries are i.i.d. drawn from a continuous distribution.
From \eqref{eq:align_block}, the following relation holds:
\begin{align}
\label{eq_blockre}
\mathbf{v}_{j,k}^{[i]} = \begin{cases}
(\boldsymbol{\Phi}\otimes \mathbf{I}_{M})\mathbf{v}_{j,S_{i}+1}^{[i]} & \mbox{if } T_{i}|L_{i} \neq 0 \mbox{ and } k = 1, \\
\mathbf{v}_{j,S_{i}+1}^{[i]} & \textrm{otherwise}. 
\end{cases}
\end{align}

{\bf Step 2} (Alignment unit): 
Next, we build an alignment unit using $U_{i}$ alignment blocks.
Specifically, $\mathbf{v}_{1+(j-1)U_{i}}^{[i]}$ through $\mathbf{v}_{jU_{i}}^{[i]}$ are used for building the $j$th  alignment unit of user $i$, denoted by $\mathbf{u}_{j}^{[i]} \in \mathbb{C}^{MT_{i}U_{i}}$ where $j = 1,\cdots, W_{i}$, which is given as 
\begin{align}
\mathbf{u}_{j}^{[i]} =\left [(\mathbf{u}_{j,1}^{[i]})^{T} (\mathbf{u}_{j,2}^{[i]})^{T} \cdots (\mathbf{u}_{j,S_{i}+1}^{[i]})^{T}\right]^{T} \label{eq:align_unit}
\end{align}
where 
\begin{align}
\mathbf{u}_{j,k}^{[i]} =
\left [
\begin{array}{c}
\mathbf{v}_{1+(j-1)U_{i},(1-k)|S_{i}+1}^{[i]} \\
\mathbf{v}_{2+(j-1)U_{i},(2-k)|S_{i}+1}^{[i]} \\
\vdots \\
\mathbf{v}_{jU_{i},(U_{i}-k)|S_{i}+1}^{[i]}
\end{array}
\right] \in \mathbb{C}^{MU} \notag 
\end{align}
for $k = 1,\cdots,S_{i}$ and
\begin{align}
\mathbf{u}_{j,S_{i}+1}^{[i]} =
\left [
\begin{array}{c}
\mathbf{v}_{1+(j-1)U_{i},S_{i}+1}^{[i]} \\
\mathbf{v}_{2+(j-1)U_{i},S_{i}+1}^{[i]} \\
\vdots \\
\mathbf{v}_{jU_{i},S_{i}+1}^{[i]}
\end{array}
\right] \in \mathbb{C}^{L_{i}MU_{i}}.
\notag
\end{align}
%

From \eqref{eq_blockre}, the following relations hold for $k = 1,\cdots,S_{i}$:
\begin{align}
\label{eq_blockre2}
\mathbf{u}_{j,k}^{[i]} =  \begin{cases}
\mathbf{u}_{j,S_{i}+1}^{[i]} & \mbox{if } T_{i}|L_{i} = 0,\\
(\mathbf{I}_{{U_{i}}/{S_{i}}} \otimes \mathbf{Q}_{k}^{[i]} \otimes \mathbf{I}_{M})\mathbf{u}_{j,S_{i}+1}^{[i]} & \textrm{otherwise}\\
\end{cases}
\end{align}
where $\mathbf{Q}_{k}^{[i]}\in \mathbb{C}^{(T_{i}-L_{i}) \times L_{i}S_{i}}$ is the block-diagonal matrix consisting of $S_{i}$ blocks whose blocks are all $\mathbf{I}_{L_{i}}$ except that the $k$th block is $\boldsymbol{\Phi}$. 
For convenience, let us call $\mathbf{u}_{j,k}^{[i]}$ as the $k$th  sub-unit  of $\mathbf{u}_{j}^{[i]}$. 



{\bf Step 3} (Transmit signal vector for user $i$): We then construct the transmit signal vector for user $i$ using $\mathbf{u}_{1}^{[i]}$ through $\mathbf{u}_{W_{i}}^{[i]}$.
The transmit signal for user $i$, denoted by $\mathbf{x}_{i}\in\mathbb{C}^{MUW + M\sum_{i \in \Lambda} L_{i}U_{i}W_{i}}$, is defined as
\begin{align}
\label{eq_x1x2}
\mathbf{x}_{i} = \left[ \mathbf{x}_{i,1}^{T} \ \mathbf{0}_{C_{1,i-1}\times 1}^{T} \ \mathbf{x}_{i,2}^{T} \ \mathbf{0}_{C_{i+1,|\Lambda|}\times 1}^{T} \right ]^{T}
\end{align}
where $C_{l,m} = \sum_{p=l}^{m}L_{p}MU_{p}W_{p}$ and $\mathbf{x}_{i,1}$ consists of $\{ \mathbf{u}_{j,k}^{[i]} \}_{k=1,\cdots,S_{i}}^{j=1,\cdots,W_{i}}$, total $W$ sub-units, and $\mathbf{x}_{i,2}$ consists of $\{ \mathbf{u}_{j,k}^{[i]} \}_{k=S_{i}+1}^{j=1,\cdots,W_{i}}$, total $W_{i}$ sub-units, defined as in the followings.
\begin{align}
\label{eq_x1x3}
\mathbf{x}_{i,1} = \left [
\begin{array}{c}
\mathbf{u}_{f^{[i]}(1)}^{[i]}
\\
\mathbf{u}_{f^{[i]}(2)}^{[i]}
\\
\vdots
\\
\mathbf{u}_{f^{[i]}(W)}^{[i]}
\end{array}
\right ] \in \mathbb{C}^{M U W},  
\ 
\mathbf{x}_{i,2} = \left [
\begin{array}{c}
\mathbf{u}_{1,S_{i}+1}^{[i]}
\\
\mathbf{u}_{2,S_{i}+1}^{[i]}
\\
\vdots
\\
\mathbf{u}_{W_{i},S_{i}+1}^{[i]}
\end{array}
\right ] \in \mathbb{C}^{ L_{i} M U_{i} W_{i}}
\end{align}
Here, $f^{[i]}$ for ${i\in\Lambda}$ is a function on $\{l \in \mathbb{N} :1 \leq l \leq W \}$ to $\{(j,k) \in \mathbb{N}^{2} : 1 \leq j \leq W_{i}, 1 \leq k \leq S_{i} \}$, defined by $f^{[i]}(l) = (f_{1}^{[i]}(l), f_{2}^{[i]}(l))$ such that
\begin{align}\label{eq:g}
&f_{1}^{[i]}(l) = ((l-1)\setminus\prod_{p=1}^{i}S_{p}) \prod_{p=1}^{i-1}S_{p} + 1 + (l-1)|\prod_{p=1}^{i-1}S_{p}, \notag
\\ 
&f_{2}^{[i]}(l) = ((l-1)|\prod_{p=1}^{i}S_{p}) \setminus \prod_{p=1}^{i-1}S_{p}   +1. 
\end{align}
The following lemma shows that every element of $\{ \mathbf{u}_{j,k}^{[i]} \}_{k=1,\cdots,S_{i}}^{j=1,\cdots,W_{i}}$ appears once in $\mathbf{x}_{i,1}$. 
\begin{lemma}\label{lemma:inverse}
Let $\mathcal{A} = \{l \in \mathbb{N} : 1 \leq l \leq W \}$, $\mathcal{B} = \{(j,k) \in \mathbb{N}^{2} : 1 \leq j \leq W_{i}, 1 \leq k \leq S_{i} \}$.
Let $f^{[i]}$ for $i\in\Lambda$ be a function on $\mathcal{A}$ to $\mathcal{B}$ defined in \eqref{eq:g} and 
let $g^{[i]}$ for $i\in\Lambda$ be a function on $\mathcal{B}$ to $\mathcal{A}$ defined by
\begin{align}
g^{[i]}(j,k) = 1 +  ((j-1)\setminus\prod_{p=1}^{i-1}S_{p}) \prod_{p=1}^{i}S_{p} +  (j-1)|\prod\limits_{p = 1}^{i-1}S_{p} + (k-1)\prod\limits_{p=1}^{i-1}S_{p}. \label{eq:f}
\end{align}
Then $g^{[i]}$ is the inverse function of $f^{[i]}$. 
\end{lemma}
\begin{IEEEproof}
We refer to Appendix \ref{properties1} for the proof.
\end{IEEEproof}


{\bf Step 4} (Transmit signal vector): Finally, transmit signal vector $\mathbf{x}^{n}$ is the sum of the transmit signal vector for each user as
\begin{align}
\label{eq_total}
\mathbf{x}^{n} = \sum\limits_{i\in\Lambda}\mathbf{x}_{i}
\end{align}
where
\begin{align}
\label{eq_duration}
n = UW + \sum_{i \in \Lambda} L_{i}U_{i}W_{i}
\end{align}
because $\mathbf{x}_{i}\in\mathbb{C}^{MUW + M\sum_{i \in \Lambda} L_{i}U_{i}W_{i}}$.
That is, at time $t=1 ,\cdots, UW + \sum_{i \in \Lambda} L_{i}U_{i}W_{i}$, the transmitter sends from the $((t-1)M+1)$th to the $(tM)$th elements of $ \sum_{i\in\Lambda}\mathbf{x}_{i}$ through $M$ antennas.

\subsubsection{Mode switching patterns at receivers} \label{subsubsec:mode_switching}
Based on the proposed transmit beamforming stated above, we design the mode switching patterns at receivers, which is fixed regardless of channel realizations.
From \eqref{eq_x1x2} and \eqref{eq_total}, we have
\begin{align}
\mathbf{x}^{n} = 
[\underbrace{(\sum_{i\in\Lambda}\mathbf{x}_{i,1})^{T}}_{\mbox{\footnotesize block}\ 1}\ \underbrace{\mathbf{x}_{1,2}^{T} \ \mathbf{x}_{2,2}^{T} \ \cdots \ \mathbf{x}_{|\Lambda|,2}^{T}}_{\mbox{\footnotesize block}\ 2} ]^{T}. \label{eq:total1}
\end{align}
Let us denote $\sum_{i\in\Lambda}\mathbf{x}_{i,1}$ and $[\mathbf{x}_{1,2}^{T} \ \mathbf{x}_{2,2}^{T} \ \cdots \ \mathbf{x}_{|\Lambda|,2}^{T}]^{T}$ in $\mathbf{x}^{n}$ as block 1 and block 2  respectively.
Subsequently, the received signal vector of user $i$ is divided as
\begin{align}
\mathbf{y}_{i}^{n} = \left [ \mathbf{y}_{i,0}^{T} \ \mathbf{y}_{i,1}^{T} \cdots \mathbf{y}_{i,|\Lambda|}^{T} \right ]^{T}
\notag
\end{align}
where $\mathbf{y}_{i,0}$ and $\mathbf{y}_{i,j}$ for $j = 1, \cdots, |\Lambda|$ are the received signal vectors induced by $\sum_{i\in\Lambda}\mathbf{x}_{i,1}$ and $\mathbf{x}_{j,2}$ respectively.
Now, we design each user's mode switching pattern during blocks 1 and 2 in the following.
For convenience, we simply call a selection pattern (of user $i$ at time $t$) to denote a specific selection matrix $\mathbf{\Gamma}_i(t)$.
We omit rigorous description of selection patterns, nonetheless one can infer them from associated channel matrices induced by selection matrices, i.e., $\mathbf{\Gamma}_i(t)\mathbf{H}_{i}$.

{\bf Mode switching pattern during block 1}:
From \eqref{eq_x1x3}, block 1 is divided as
\begin{align}
\sum_{i\in\Lambda}\mathbf{x}_{i,1} = \left [
\begin{array}{c}
\sum\limits_{i \in \Lambda}\mathbf{u}_{f^{[i]}(1)}^{[i]} \\
\vdots \\
\sum\limits_{i \in \Lambda}\mathbf{u}_{f^{[i]}(W)}^{[i]}
\end{array}
\right ]. \label{eq:block1}
\end{align}

Note that the time interval for transmitting block 1 is 
\begin{align}\label{eq:time_block1}
1\leq t\leq UW.
\end{align}
During block 1, user $i$ exploits a set of $S_{i}$ selection patterns repeatedly over the entire time interval in \eqref{eq:time_block1}.
The channel matrix associated with the $j$th selection pattern, denoted by $\mathbf{H}_{i,j} \in \mathbb{C}^{L_{i} \times M}$, is given by 
\begin{align}
\begin{array}{cc}
\mathbf{H}_{i,j} = \left [ 
\begin{array}{c}
\mathbf{h}_{i,1+(j-1)L_{i}} \\
\mathbf{h}_{i,2+(j-1)L_{i}} \\
\vdots \\
\mathbf{h}_{i,jL_{i}}
\end{array}
\right ] & j = 1,\cdots, S_{i} \\
\end{array} \notag
\end{align}
where $\mathbf{h}_{k,l}\in \mathbb{C}^{1\times M}$ is the $l$th row vector of $\mathbf{H}_{k}$ for $k\in\Lambda$ and $l = 1,\cdots, N_{k}$.
%
For this case, at each time instant,  each user chooses the selection pattern of which index is same as that of the currently transmitted  sub-unit  of his transmit signal vector.
One can see from \eqref{eq:block1} that the sub-unit of user $i$ transmitted at time $t=1,\cdots,UW$ is given by $\mathbf{u}_{f^{[i]}(l(t))}^{[i]}$ where $l(t) = 1 + (t-1)\setminus U$.
Then, at time $t=1,\cdots,UW$, the user $i$ receives the transmit signal vector using the $f_{2}^{[i]}(l(t))$th selection pattern, associated with $\mathbf{H}_{i,f_{2}^{[i]}(l(t))}$.
As a result, the received signal vector of user $i$ during $t=1,\cdots,UW$ is given by
\begin{align} \label{eq:y_i0}
\mathbf{y}_{i,0}= 
\left [
\begin{array}{c}
(\mathbf{I}_{U} \otimes \mathbf{H}_{i,f_{2}^{[i]}(1)} )\sum\limits_{j \in \Lambda}\mathbf{u}_{f^{[j]}(1)}^{[j]}
\\
(\mathbf{I}_{U} \otimes \mathbf{H}_{i,f_{2}^{[i]}(2)} )\sum\limits_{j \in \Lambda}\mathbf{u}_{f^{[j]}(2)}^{[j]}
\\
\vdots
\\
(\mathbf{I}_{U} \otimes \mathbf{H}_{i,f_{2}^{[i]}(W)} ) \sum\limits_{j \in \Lambda}\mathbf{u}_{f^{[j]}(W)}^{[j]}
\end{array} \right ]. 
\end{align}

{\bf Mode switching pattern during block 2}:
We divide block 2 into desired signal and interference signal parts of user $i$, in which
desired signal part is $\mathbf{x}_{i,2}$  and interference signal part is the rest of block 2 except $\mathbf{x}_{i,2}$.
Note that the time interval for transmitting $\mathbf{x}_{i,2}$ is 
\begin{align}\label{eq:time_block2_desired}
a_{i}+1 \leq t \leq a_{i+1}
\end{align}
where $a_{i} = UW + \sum_{p=1}^{i-1}L_{p}U_{p}W_{p}$.
Let us define $\mathbf{H}_{i,S_{i}+1} \in \mathbb{C}^{(T_{i} - L_{i}S_{i}) \times M}$ and $\mathbf{H}_{i,j,k} \in \mathbb{C}^{(L_{i}-T_{i}|L_{i}) \times M}$ as
\begin{align} \label{eq:H_i_S}
\mathbf{H}_{i,S_{i}+1} = \left [ 
\begin{array}{c}
\mathbf{h}_{i,L_{i}S_{i}+1} \\
\mathbf{h}_{i,L_{i}S_{i}+2} \\
\vdots \\
\mathbf{h}_{i,T_{i}}
\end{array}
\right ] 
\end{align} 
and
\begin{align} \label{eq:H_ijk}
\mathbf{H}_{i,j,k} = 
\left [ 
\begin{array}{c}
\mathbf{h}_{i, (j-1) L_{i} + (k-1)|L_{i}+1} \\
\mathbf{h}_{i, (j-1) L_{i} + k|L_{i}+1} \\
\vdots \\
\mathbf{h}_{i, (j-1) L_{i} + (k-2 + L_{i}-(T_{i}|L_{i}))|L_{i}+1}
\end{array}
\right ] \ 
\end{align}
for $j = 1,\cdots, S_{i}$ and $k = 1,\cdots, L_{i}$.

First consider the desired signal part of user $i$ during block 2.
For the transmission of $\mathbf{x}_{i,2}$, the mode switching pattern of user $i$ differs according to the value of $T_{i}|L_{i}$.
If $T_{i}|L_{i}=0$, then user $i$ exploits a single selection pattern repeatedly over the entire time interval in \eqref{eq:time_block2_desired}. 
The channel matrix associated with the selection pattern is given by \eqref{eq:H_i_S}, where $\mathbf{H}_{i,S_{i}+1}  \in \mathbb{C}^{L_{i} \times M}$ in this case.
%
If $T_{i}|L_{i}\neq 0$, then 
 user $i$ exploits a set of $L_{i}S_{i}$ selection patterns repeatedly over the entire time interval in \eqref{eq:time_block2_desired}, i.e., the number of $L_{i}U_{i}W_{i}/(L_{i}S_{i}) = \prod_{p \in \Lambda \setminus \{i\}} S_{p} (T_{p}-L_{p})$ repetitions.
 The channel matrix associated with the $j$th selection pattern, denoted by $\mathbf{H}_{i,S_{i}+1,j}  \in \mathbb{C}^{L_{i} \times M}$ for $j = 1 , \cdots ,  L_{i}S_{i}  $ which can be constructed from \eqref{eq:H_i_S} and \eqref{eq:H_ijk}, is given by 
\begin{align}
\label{eq_chan}
\mathbf{H}_{i,S_{i} + 1,j} = 
\left [ 
\begin{array}{c}
\mathbf{H}_{i,S_{i}+1}
\\
\mathbf{H}_{i,(j-1)\setminus L_{i}+1, (j-1)|L_{i} + 1}
\end{array}
\right ].
\end{align}
Then, at time $t=a_i+1,\cdots,a_{i+1}$, user $i$ receives the transmit signal vector using the selection pattern corresponding to $\mathbf{H}_{i,S_{i}+1,l_{i}(t)}$ where
$l_{i}(t) =1 +  (t-a_{i}-1)|(L_{i}S_{i})$.
As a result, the received signal vector of user $i$ during $t=a_i+1,\cdots,a_{i+1}$ is given by
\begin{align}
\mathbf{y}_{i,i}= 
\left [
\begin{array}{c}
( \mathbf{I}_{U_{i}/S_{i}} \otimes \mathbf{H}_{i,S_{i}+1}')\mathbf{u}_{1,S_{i}+1}^{[i]}
\\
( \mathbf{I}_{U_{i}/S_{i}} \otimes \mathbf{H}_{i,S_{i}+1}')\mathbf{u}_{2,S_{i}+1}^{[i]}
\\
\vdots
\\
( \mathbf{I}_{U_{i}/S_{i}} \otimes \mathbf{H}_{i,S_{i}+1}')\mathbf{u}_{W_{i},S_{i}+1}^{[i]}
\end{array}
\right ] \label{eq_rs3_desired}
\end{align}
where 
\begin{align}
\mathbf{H}_{i,S_{i}+1}' = \begin{cases}
\mathbf{I}_{L_{i}S_{i}} \otimes \mathbf{H}_{i,S_{i}+1} & \mbox{if } T_{i}|L_{i}=0, \\
\operatorname{diag}\left ( \mathbf{H}_{i,S_{i}+1,1},\mathbf{H}_{i,S_{i}+1,2}, \cdots, \mathbf{H}_{i,S_{i}+1,L_{i}S_{i}} \right ) & \textrm{otherwise}.
\end{cases}
\notag
\end{align}

Now consider the interference signal part of user $i$ during block 2.
From \eqref{eq:time_block2_desired}, the time interval for transmitting $\mathbf{x}_{i',2}$, where $i' \in \Lambda\setminus \{i\}$, is given by
\begin{align} \label{eq:time_block2_intf}
a_{i'}+1 \leq t \leq a_{i'+1}.
\end{align}
For the transmission of  $\mathbf{x}_{i',2}$, user $i$ exploits the same set of $S_{i}$ selection patterns used for block 1 again over the entire time interval in \eqref{eq:time_block2_intf}.
For this case, at each time instant, each user chooses the selection pattern of which index is the same as that used to receive the first sub-unit of the alignment unit to which the currently transmitted sub-unit belongs.
Specifically, from \eqref{eq_x1x3}, the sub-unit of user $i'$ transmitted at time $t= a_{i'}+1 ,\cdots, a_{i'+1}$ is given by $\mathbf{u}_{l(t),S_{i'}+1}^{[i']}$ where $l(t) = 1+ (t-1-a_{i'}) \setminus (L_{i'}U_{i'})$. 
Since $\mathbf{u}_{l(t),1}^{[i']} = \mathbf{u}_{f^{[i']}(g^{[i']}(l(t),1))}^{[i']}$ from Lemma \ref{lemma:inverse},
$\mathbf{u}_{l(t),1}^{[i']}$ is a summand of $\sum_{j \in \Lambda}\mathbf{u}_{f^{[j]}(g^{[i']}(l(t),1))}^{[j]}$, 
which means from \eqref{eq:block1} that $\mathbf{u}_{l(t),1}^{[i']}$ is transmitted simultaneously with $\mathbf{u}_{f^{[i]}(g^{[i']}(l(t),1))}^{[i]}$ in block 1. That is, user $i$ exploits the $f_{2}^{[i]}(g^{[i']}(l(t),1))$th selection pattern to receive $\mathbf{u}_{f^{[i]}(g^{[i']}(l(t),1))}^{[i]}$
so that, at time $t= a_{i'}+1 ,\cdots, a_{i'+1}$, user $i$ receives the transmit signal vector using the $f_{2}^{[i]}(g^{[i']}(l(t),1))$th selection pattern, associated with $\mathbf{H}_{i,f_{2}^{[i]}(g^{[i']}(l(t),1))}$.
As a result, the received signal vector of user $i$ induced by $\mathbf{x}_{i',2}$ is
\begin{align}
\label{eq_rs3}
\mathbf{y}_{i,i'} = 
\left [
\begin{array}{c}
(\mathbf{I}_{L_{i}U_{i}} \otimes \mathbf{H}_{i,f_{2}^{[i]}(g^{[i']}(1,1))})\mathbf{u}_{1,S_{i'}+1}^{[i']}
\\
(\mathbf{I}_{L_{i}U_{i}} \otimes \mathbf{H}_{i,f_{2}^{[i]}(g^{[i']}(2,1))})\mathbf{u}_{2,S_{i'}+1}^{[i']}
\\
\vdots
\\
(\mathbf{I}_{L_{i}U_{i}} \otimes \mathbf{H}_{i,f_{2}^{[i]}(g^{[i']}(W_{i'},1))} ) \mathbf{u}_{W_{i'},S_{i'}+1}^{[i']}
\end{array} \right ] .
\end{align}

\subsubsection{Interference cancellation at receivers}
In the following, we show that user $i$ can eliminate all interference signals contained in $\mathbf{y}_{i,0}$ in \eqref{eq:y_i0} using the received interference signal parts during block 2, i.e., $\mathbf{y}_{i,i'}$ in \eqref{eq_rs3} for all $i'\in\Lambda\setminus\{i\}$. 
First, we introduce the following lemma, which plays a key role to verify such interference cancellation.
\begin{lemma}\label{lemma:const}
Let $\mathcal{A} = \{l \in \mathbb{N} : 1 \leq l \leq W \}$, $\mathcal{B} = \{(j,k) \in \mathbb{N}^{2} : 1 \leq j \leq W_{i}, 1 \leq k \leq S_{i} \}$.
Let $f^{[i]}$ for $i\in\Lambda$ be a function on $\mathcal{A}$ to $\mathcal{B}$ defined in \eqref{eq:g} and 
let $g^{[i]}$ for $i\in\Lambda$ be a function on $\mathcal{B}$ to $\mathcal{A}$ defined in \eqref{eq:f}.
For $i,i' \in \Lambda$ where $i\neq i'$ and $(j,k), (j,k') \in \mathcal{B}$, the following relation holds:
\begin{align}f_{2}^{[i]}(g^{[i']}(j,k)) = f_{2}^{[i]}(g^{[i']}(j,k')) \notag
\end{align}
\end{lemma}
\begin{IEEEproof}
We refer to Appendix \ref{properties2} for the proof.
\end{IEEEproof}

Consider $(\mathbf{I}_{U} \otimes \mathbf{H}_{i,f_{2}^{[i]}(l)} )\mathbf{u}_{f^{[i']}(l)}^{[i']}$ for $i'\in\Lambda\setminus\{i\}$ and $1\leq l \leq W$, which is an interference vector in $\mathbf{y}_{i,0}$.
We have
\begin{align} \label{eq_rsequ}
(\mathbf{I}_{U} \otimes \mathbf{H}_{i,f_{2}^{[i]}(l)} )\mathbf{u}_{f^{[i']}(l)}^{[i']} 
&= (\mathbf{I}_{U} \otimes \mathbf{H}_{i,f_{2}^{[i]}(g^{[i']}(f^{[i']}(l)))})\mathbf{u}_{f^{[i']}(l)}^{[i']}\notag \\
& = (\mathbf{I}_{U} \otimes \mathbf{H}_{i,f_{2}^{[i]}(g^{[i']}(f_{1}^{[i']}(l),1))})\mathbf{u}_{f^{[i']}(l)}^{[i']} 
\end{align}
where the first and second equalities follow from Lemma \ref{lemma:inverse} and Lemma \ref{lemma:const} respectively.
If $T_{i'}|L_{i'} = 0$, then, substituting \eqref{eq_blockre2} into \eqref{eq_rsequ}, we have 
\begin{align}
\label{eq_inf1}
(\mathbf{I}_{U} \otimes \mathbf{H}_{i,f_{2}^{[i]}(l)} )\mathbf{u}_{f^{[i']}(l)}^{[i']} 
=(\mathbf{I}_{L_{i}U_{i}} \otimes \mathbf{H}_{i,f_{2}^{[i]}(g^{[i']}(f_{1}^{[i']}(l),1))})\mathbf{u}_{f_{1}^{[i']}(l),S_{i'}+1}^{[i']}
\end{align}
where $U = L_{i}U_{i}$ for this case.
If $T_{i'}|L_{i'} \neq 0$, then, substituting \eqref{eq_blockre2} into \eqref{eq_rsequ}, we have
\begin{align}
(\mathbf{I}_{U} \otimes \mathbf{H}_{i,f_{2}^{[i]}(l)} )\mathbf{u}_{f^{[i']}(l)}^{[i']} 
&=(\mathbf{I}_{U} \otimes \mathbf{H}_{i,f_{2}^{[i]}(g^{[i']}(f_{1}^{[i']}(l),1))}) (\mathbf{I}_{{U_{i}}/{S_{i}}} \otimes \mathbf{Q}_{f_{2}^{[i']}(l)}^{[i']}  \otimes \mathbf{I}_{M})\mathbf{u}_{f_{1}^{[i']}(l),S_{i'}+1}^{[i']} \notag \\
& = (\mathbf{I}_{{U_{i}}/{S_{i}}} \otimes \mathbf{Q}_{f_{2}^{[i']}(l)}^{[i']}\otimes \mathbf{I}_{L_{i}})(\mathbf{I}_{L_{i}U_{i}} \otimes \mathbf{H}_{i,f_{2}^{[i]}(g^{[i']}(f_{1}^{[i']}(l),1))})\mathbf{u}_{f_{1}^{[i']}(l),S_{i'}+1}^{[i']} .
\label{eq_inf2}
\end{align}
Here \eqref{eq_inf2} comes from the following relation:
\begin{align}
\label{eq_inter}
& (\mathbf{I}_{U} \otimes \mathbf{H}_{i,f_{2}^{[i]}(g^{[i']}(f_{1}^{[i']}(l),1))})(\mathbf{I}_{{U_{i}}/{S_{i}}} \otimes \mathbf{Q}_{f_{2}^{[i']}(l)}^{[i']}  \otimes \mathbf{I}_{M}) \notag\\
& = (\mathbf{I}_{{U_{i}}/{S_{i}}} \otimes \mathbf{Q}_{f_{2}^{[i']}(l)}^{[i']}  \otimes \mathbf{H}_{i,f_{2}^{[i]}(g^{[i']}(f_{1}^{[i']}(l),1))})  \notag\\
& = ((\mathbf{I}_{{U_{i}}/{S_{i}}} \otimes \mathbf{Q}_{f_{2}^{[i']}(l)}^{[i']}) \mathbf{I}_{L_{i}U_{i}})  
\otimes ( \mathbf{I}_{L_{i}} \mathbf{H}_{i,f_{2}^{[i]}(g^{[i']}(f_{1}^{[i']}(l),1))})  \notag\\
& = (\mathbf{I}_{{U_{i}}/{S_{i}}} \otimes \mathbf{Q}_{f_{2}^{[i']}(l)}^{[i']}\otimes \mathbf{I}_{L_{i}})(\mathbf{I}_{L_{i}U_{i}} \otimes \mathbf{H}_{i,f_{2}^{[i]}(g^{[i']}(f_{1}^{[i']}(l),1))})
\end{align}
where \eqref{eq_inter} follows from the mixed-product property that for matrices $\mathbf{A}$, $\mathbf{B}$, $\mathbf{C}$, and $\mathbf{D}$ in which the matrix products $\mathbf{AC}$ and $\mathbf{BD}$ can be defined, $(\mathbf{A} \otimes \mathbf{B}) (\mathbf{C} \otimes \mathbf{D}) = \mathbf{AC} \otimes \mathbf{BD}$, see \cite[Lemma 4.2.10]{horn1991topics}.

From \eqref{eq_rs3}, user $i$ is able to extract the following vector from $\mathbf{y}_{i,i'}$:
\begin{align}
\label{eq_inf4}
(\mathbf{I}_{L_{i}U_{i}} \otimes \mathbf{H}_{i,f_{2}^{[i]}(g^{[i']}(f_{1}^{[i']}(l),1))})\mathbf{u}_{f_{1}^{[i']}(l),S_{i'}+1}^{[i']}.
\end{align}
Then, user $i$ constructs $(\mathbf{I}_{U} \otimes \mathbf{H}_{i,f_{2}^{[i]}(l)} )\mathbf{u}_{f^{[i']}(l)}^{[i']}$ using \eqref{eq_inf4} from the relations in \eqref{eq_inf1} and \eqref{eq_inf2} and subtracts it from $\mathbf{y}_{i,0}$.
In the same manner, user $i$ can remove all interference vectors in $\mathbf{y}_{i,0}$.

\subsubsection{Achievable LDoF}
Let us denote the remaining signal vector after cancelling all interference vectors in $\mathbf{y}_{i,0}$ as $\mathbf{y}_{i,0}'$.
Combining $\mathbf{y}_{i,0}'$ with $\mathbf{y}_{i,i}$, user $i$ has
\begin{align} \label{eq_result}
\left [
\begin{array}{c}
\mathbf{y}_{i,0}' \\
\mathbf{y}_{i,i} \\
\end{array}
\right ]
=
\left [
\begin{array}{c}
(\mathbf{I}_{U} \otimes \mathbf{H}_{i,f_{2}^{[i]}(1)} )\mathbf{u}_{f^{[i]}(1)}^{[i]}
\\
\vdots
\\
(\mathbf{I}_{U} \otimes \mathbf{H}_{i,f_{2}^{[i]}(\prod_{p \in \Lambda} S_{p})} ) \mathbf{u}_{f^{[i]}(\prod_{p \in \Lambda} S_{p})}^{[i]}
\\
( \mathbf{I}_{U_{i}/S_{i}} \otimes \mathbf{H}_{i,S_{i}+1}')\mathbf{u}_{1,S_{i}+1}^{[i]}
\\
\vdots
\\
( \mathbf{I}_{U_{i}/S_{i}} \otimes \mathbf{H}_{i,S_{i}+1}')\mathbf{u}_{W_{i},S_{i}+1}^{[i]}
\end{array} \right ].
\end{align}
By classifying \eqref{eq_result} by  alignment units,  \eqref{eq_result2} is decomposed into $W_{i}$ segments as follows:
\begin{align}
\label{eq_result2}
\left [
\begin{array}{c}
(\mathbf{I}_{U} \otimes \mathbf{H}_{i,1} )\mathbf{u}_{j,1}^{[i]}
\\
\vdots
\\
(\mathbf{I}_{U} \otimes \mathbf{H}_{i,S_{i}} )\mathbf{u}_{j,S_{i}}^{[i]} 
\\
( \mathbf{I}_{U_{i}/S_{i}} \otimes \mathbf{H}_{i,S_{i}+1}')\mathbf{u}_{j,S_{i}+1}^{[i]}
\end{array} \right ] \ \ j = 1, \cdots , W_{i}.
\end{align}
By classifying \eqref{eq_result2} by  alignment blocks,  \eqref{eq_result} is further decomposed into $U_{i}W_{i}$ segments as follows:
\begin{align}
\label{eq_result3}
\left [
\begin{array}{c}
(\mathbf{I}_{N_{i}} \otimes \mathbf{H}_{i,(j-1)|S_{i}+1} )\mathbf{v}_{j,1}^{[i]}
\\
(\mathbf{I}_{N_{i}} \otimes \mathbf{H}_{i,(j-2)|S_{i}+1} )\mathbf{v}_{j,2}^{[i]}
\\
\vdots
\\
(\mathbf{I}_{N_{i}} \otimes \mathbf{H}_{i,(j-S_{i})|S_{i}+1} )\mathbf{v}_{j,S_{i}}^{[i]} 
\\
\mathbf{H}_{i,S_{i}+1,j}'\mathbf{v}_{j,S_{i}+1}^{[i]}
\end{array} \right ] \ \ j = 1, \cdots , U_{i}W_{i}
\end{align}
where  
\begin{align}
\mathbf{H}_{i,S_{i}+1,j}' =
\begin{cases}
\mathbf{I}_{L_{i}} \otimes \mathbf{H}_{i,S_{i}+1} & \mbox{if } T_{i}|L_{i} = 0, \\
\operatorname{diag}\left( \mathbf{H}_{i,S_{i+1},((j-1)|S_{i})L_{i}+1}, \mathbf{H}_{i,S_{i+1},((j-1)|S_{i})L_{i}+2}, \cdots, \mathbf{H}_{i,S_{i+1},((j-1)|S_{i}+1)L_{i}} \right) & \textrm{otherwise}.
\end{cases} \notag
\end{align}

If $T_{i}|L_{i}=0$, then, substituting \eqref{eq:align_block} into \eqref{eq_result3} and switching the rows, we have
\begin{align}
\label{eq_result4}
& (\mathbf{I}_{L_{i}} \otimes  \mathbf{H}_{i} ) (\mathbf{I}_{L_{i}} \otimes \mathbf{I}_{M,T_{i}}) \mathbf{s}_{j}^{[i]} = (\mathbf{I}_{L_{i}} \otimes  [\mathbf{H}_{i}]_{T_{i}} )\mathbf{s}_{j}^{[i]} 
\end{align}
where $[\mathbf{H}_{i}]_{T_{i}}$ is the leading principal minor of $\mathbf{H}_{i}$ of order $T_{i}$.
Since $(\mathbf{I}_{L_{i}} \otimes  [\mathbf{H}_{i}]_{T_{i}} )$ is non-singular almost surely, user $i$ can obtain $\mathbf{s}_{j}^{[i]}$ from \eqref{eq_result4} almost surely.

If $T_{i}|L_{i} \neq 0$, then, substituting \eqref{eq:align_block} into \eqref{eq_result3}, we have
\begin{align}
\label{eq_result5}
\underbrace{
\left [
\begin{array}{c}
\boldsymbol{\Phi} \otimes \mathbf{H}_{i,(j-1)|S_{i}+1} 
\\
\mathbf{I}_{N_{i}} \otimes \mathbf{H}_{i,(j-2)|S_{i}+1} 
\\
\vdots
\\
\mathbf{I}_{N_{i}} \otimes \mathbf{H}_{i,(j-S_{i})|S_{i}+1}  
\\
\mathbf{H}_{i,S_{i}+1,j}' 
\end{array} \right ]}_{\mathbf{B}_{j}^{[i]}}  (\mathbf{I}_{L_{i}} \otimes \mathbf{I}_{M,T_{i}}) \mathbf{s}_{j}^{[i]} .
\end{align}
It can be easily verified that $\mathbf{B}_{j}^{[i]} = \mathbf{C}_{j}^{[i]} (\mathbf{I}_{L_{i}} \otimes  [\mathbf{H}_{i}]_{T_{i}} )$ where $\mathbf{C}_{j}^{[i]} \in \mathbb{C}^{L_{i}T_{i} \times L_{i}T_{i}}$ for $j = 1, \cdots, U_{i}W_{i}$ is a non-singular matrix almost surely so that  user $i$ can obtain $\mathbf{s}_{j}^{[i]}$ from \eqref{eq_result5} almost surely.
Consequently, user $i$ is able to $\mathbf{s}_{1}^{[i]}$ through $\mathbf{s}_{U_{i}W_{i}}^{[i]}$ almost surely.

Since total $L_{i}T_{i}U_{i}W_{i}$ independent information symbols are delivered almost surely to user $i$ during the period given in \eqref{eq_duration}, the achievable LDoF of user $i$ is given by 
\begin{align}
d_{i}= & \ \frac{\frac{T_{i}L_{i}}{T_{i}-L_{i}}\prod\limits_{p \in \Lambda} S_{p}(T_{p} - L_{p})}
{\prod\limits_{p \in \Lambda} S_{p}(T_{p} - L_{p}) + \sum\limits_{p \in \Lambda} \frac{L_{p}}{T_{p}-L_{p}}\prod\limits_{q \in \Lambda} S_{q}(T_{q} - L_{q}) }  \\
=& \ 
\frac{\frac{T_{i}L_{i}}{T_{i}-L_{i}}}
{1 + \sum\limits_{p \in \Lambda} \frac{L_{p}}{T_{p}-L_{p}}} .
\end{align}
Therefore,
\begin{align} \label{eq:dof3_2}
d_{\Sigma}=\sum\limits_{i=1}^{K}d_{i} = \frac{\sum\limits_{i\in\Lambda}\frac{T_{i}L_{i}}{T_{i}-L_{i}}}
{1 + \sum\limits_{i \in \Lambda} \frac{L_{i}}{T_{i}-L_{i}}} 
\end{align}
is achievable for Case 3-2.
In conclusion, from \eqref{eq:dof_total}, \eqref{eq:dof1_2}, and \eqref{eq:dof3_2}, $d_{\Sigma}=\min( M , \max ( L_{\max} , \eta ) )$ is achievable, which completes the achievability proof of Theorem \ref{main_thm}.

\subsection{Achievability of Theorem \ref{main_thm2}} \label{achievability_thm2}
For notational convenience, we define the inequality in \eqref{eq:ineq1} as $I_{1k}$ and the inequality $d_{k} \geq 0$ as $I_{2k}$ in the rest of this subsection.
Since the LDoF region $\mathcal{D}$ in Theorem \ref{main_thm2} is a polyhedron, it suffices to show that all vertices of $\mathcal{D}$ are achievable. Hence, our achievability proof begins with characterizing vertices of $\mathcal{D}$.
The following lemma establishes a condition for a $K$-tuple in $\mathbb{R}^{K}$ to be a vertex of $\mathcal{D}$.
\begin{lemma}
\label{active}
Consider the LDoF region $\mathcal{D}$ in Theorem \ref{main_thm2}.
If $\mathbf{d} \in \mathbb{R}^{K}$ is a vertex of $\mathcal{D}$, then only $K$ inequalities among $\{I_{1k}, I_{2k} \}_{k \in\mathcal{K}}$ should be active
\footnote{An inequality $f(\mathbf{x}) \leq 0 $ is said to be active at $\mathbf{x}^{*}$ if $f(\mathbf{x}^{*}) = 0$ \cite[Definition 20.1]{chong2013introduction}.}
 at $\mathbf{d}$ while $I_{1k}$ and $I_{2k}$ for $k \in \mathcal{K}$ cannot be active simultaneously at $\mathbf{d}$.

\end{lemma}

\begin{IEEEproof}
Assume that $I_{11}$ and $I_{21}$ are active at $\mathbf{d} = (d_{1},\cdots,d_{K}) \in \mathcal{D}$. Combining $I_{11}$ and $I_{21}$, it is followed by
\begin{align}
\sum\limits_{i=1}^{K} \frac{d_{i}}{T_{i}}=1 \notag .
\end{align}
Hence, $\mathbf{d}$ must not be a zero vector and we can find an index $k^{*}\in\mathcal{K}$ such that $I_{2k^{*}}$ is not active at $\mathbf{d}$, i.e., $d_{k^{*}} > 0$.
Then, we have
\begin{align}
\frac{d_{k^{*}}}{L_{k^{*}}}  + \sum\limits_{i =1, i\neq  k^{*}}^{K} \frac{d_{i}}{T_{i}} = 1 + (\frac{1}{L_{k^{*}}} -\frac{1}{T_{k^{*}}}) d_{k^{*}} > 1 \notag,
\end{align}
which means that $\mathbf{d}$ does not satisfy $I_{1k^{*}}$ so that $\mathbf{d} \notin \mathcal{D}$. Contradicting the assumption, $I_{1k}$ and $I_{2k}$ for $k \in \mathcal{K}$ cannot be active simultaneously at $\mathbf{d}$ and,
as a result, for $\mathbf{d} \in \mathcal{D}$, at most $K$ inequalities among $\{I_{1k}, I_{2k} \}_{k \in\mathcal{K}}$ are active at $\mathbf{d}$.
Furthermore, if $\mathbf{d}$ is a vertex of $\mathcal{D}$, then at least $K$ inequalities among $\{I_{1k}, I_{2k} \}_{k \in\mathcal{K}}$ should be active on $\mathbf{d}$ because a vertex of a polyhedron is expressed as an intersection of at least $K$ faces of the polyhedron.
Therefore, only $K$ inequalities among $\{I_{1k}, I_{2k} \}_{k \in\mathcal{K}}$ should be active at $\mathbf{d}$, which completes the proof of Lemma \ref{active}.
\end{IEEEproof}

Consider a $K$-tuple $\mathbf{d} = (d_{1},\cdots,d_{K}) \in \mathbb{R}^{K}$ such that $K$ inequalities among $\{I_{1k}, I_{2k} \}_{k \in\mathcal{K}}$ are active at $\mathbf{d}$ while $I_{1k}$ and $I_{2k}$ for $k \in \mathcal{K}$ are not active simultaneously at $\mathbf{d}$ and
let $\Lambda_{i}=\{k \in \mathcal{K} :  I_{ik} \mbox{ is active at } \mathbf{d} \}$ for $i = 1, 2$.
Note that, from Lemma \ref{active}, $\{ \Lambda_{1}$, $\Lambda_{2}\}$ is a partition of $\mathcal{K}$. 
Assume $\Lambda_{1}=\{1,\cdots,J\}$ and $\Lambda_{2}=\{J+1,\cdots,K\}$ without loss of generality. 
Composing the $K$ inequalities active at $\mathbf{d}$, we have
\begin{align}
\mathbf{A}_{1} [d_{1} \cdots d_{J}]^{T} + \left [ \frac{1}{T_{J+1}} \mathbf{1}_{J\times 1} \ \cdots \  \frac{1}{T_{K}} \mathbf{1}_{J\times 1} \right ][d_{J+1} \cdots d_{K}]^{T}  & = \mathbf{1}_{J\times1} \notag \\
[d_{J+1} \cdots d_{K}]^{T} & = \mathbf{0}_{(K-J)\times 1} \label{eq:inequalities_active}
\end{align}
where
\begin{align}
\begin{array}{l}
\mathbf{A}_{1} = 
\left[ 
\begin{array}{c}
\begin{array}{ccccc}
\frac{1}{N_{1}} & \frac{1}{T_{2}} & \frac{1}{T_{3}} & \cdots & \frac{1}{T_{J}} \\
\frac{1}{T_{1}} & \frac{1}{N_{2}} & \frac{1}{T_{3}} & \cdots & \frac{1}{T_{J}} \\
\vdots & \vdots & \vdots & \ddots & \vdots  \\
\frac{1}{T_{1}} & \frac{1}{T_{2}} & \frac{1}{T_{3}} & \cdots & \frac{1}{N_{J}} \\
\end{array}
\end{array}
\right]
\end{array} \notag .
\end{align}
Since $\mathbf{A}_{1}$ is non-singular from Lemma \ref{det} in Appendix \ref{proof_upper}, from \eqref{eq:inequalities_active}, we have
\begin{align}
\mathbf{d} = [(\mathbf{A}_{1}^{-1} \mathbf{1}_{J \times 1})^{T} \ \mathbf{0}_{(K-J) \times 1}^{T}]^{T} \notag
\end{align}
From \eqref{eq:eta}, $\mathbf{A}_{1}^{-1} \mathbf{1}_{J \times 1}$ can be calculated easily, which results that 
\begin{align}
\label{eq_dof}
d_{i} = 
\begin{cases}
\frac{\frac{T_{i}N_{i}}{T_{i}-N_{i}}}{1+\sum\limits_{k\in\Lambda_{1}}\frac{N_{k}}{T_{k}-N_{k}}} & \mbox{if } i \in \Lambda_{1}, \\
0 & \textrm{otherwise}.
\end{cases}
\end{align}
Since \eqref{eq_dof} is achievable by supporting users in $\Lambda_{1}$ with the scheme proposed in Case 3-2 of Section \ref{achievability_thm1}, 
all the vertices of $\mathcal{D}$ are achievable, which completes the achievability proof of Theorem \ref{main_thm2}.

\subsection{Achievability of Corollary \ref{main_thm3}} \label{achievability_thm3}
If $L_{\max} \geq \eta_{\text IC}$, then it is achievable by supporting an user with the maximum number of RF-chains. 
For the rest of this section we prove that $d_{\Sigma,IC} = \eta_{\text IC}$ assuming $L_{\max} < \eta_{\text IC}$.
It can be shown that \eqref{eq:dof_ic} is achievable by modifying the achievable scheme derived for Case 3-2 in Section \ref{achievability_thm1} as follows. 
Transmitter $k\in \{ k \in \mathcal{K} : N_{k} > L_{\max} \}$ constructs transmit signal vector for user $k$ in accordance with Step 1 through Step 3 in Section \ref{subsubsec:transmit_signal} by setting $M = M_{k}$ and $T_{k} = N_{k}$ and sends it as in Step 4 in Section \ref{subsubsec:transmit_signal}.
Note that since the beamforming strategy in Section \ref{subsubsec:transmit_signal} does not require cooperation among transmitters, it can be directly applied to interference channels.
Then each user receives and decodes the transmitted signal in accordance with the procedure in Section \ref{subsubsec:mode_switching}.
Note that it can be easily shown that \eqref{eq:dof_ic} is achievable, which completes the achievability proof of Corollary \ref{main_thm3}.

\section{Concluding Remarks}\label{sec:conclusion}
In this paper, the DoF of the $K$-user MIMO BC with reconfigurable antennas under no CSIT has been studied.
We completely characterized the sum LDoF of the $K$-user MIMO BC with reconfigurable antennas under general antenna configurations and further characterized the LDoF region for a class of antenna configurations.
Our results provide a comprehensive understanding of reconfigurable antennas on the LDoF of the $K$-user MIMO BC, which demonstrates that reconfigurable antennas are beneficial for a broad class of antenna configurations. 
In particular, the DoF gain from reconfigurable antennas enlarges as both the number of transmit antennas and the number of preset modes increase.
Our analysis has been further extended to characterizing the sum LDoF of the $K$-user MIMO IC with reconfigurable antennas for a class of antenna configuration, which leads to similar argument for the $K$-user MIMO BC with reconfigurable antennas.

\appendices

\section{Proof of Technical Lemmas}

\subsection{Proof of Lemma \ref{lemma_sel}}
\label{proof_sel}

Let us define $\mathbf{G}_{1}^{c}\in \mathbb{C}^{(N_1-L_1)\times M}$ as the submatrix consisting of the ($L_{1}+1$)th through the $N_1$th rows of $\mathbf{H}_{1}$.
We will prove Lemma \ref{lemma_sel} for given realization of $\mathbf{G}_{1}$ and $[\mathbf{V}_{2}^{n} \cdots \mathbf{V}_{K}^{n}]$.
That is, `almost sure' in the rest of the proof is due to the randomness of $\mathbf{G}_1^c$.
Since Lemma \ref{lemma_sel} trivially holds if $\mathbf{\Gamma}_{1}^{n}\mathbf{H}_{1}^{n} = \mathbf{G}_{1}^{n}$ or $\mathbf{G}_{1}^{n}[\mathbf{V}_{2}^{n} \cdots \mathbf{V}_{K}^{n}] = \mathbf{0}$ , we assume $\boldsymbol{\Gamma}_{1}^{n}\mathbf{H}_{1}^{n} \neq \mathbf{G}_{1}^{n}$ and $\mathbf{G}_{1}^{n}[\mathbf{V}_{2}^{n} \cdots \mathbf{V}_{K}^{n}] \neq \mathbf{0}$ from now on.

For convenience, denote $\operatorname{rank}(\mathbf{G}_{1}^{n}[\mathbf{V}_{2}^{n} \cdots \mathbf{V}_{K}^{n}])=r\geq 1$.
Define the set of column indices consisting of $r$  linearly independent columns of $\mathbf{G}_{1}^{n}[\mathbf{V}_{2}^{n} \cdots \mathbf{V}_{K}^{n}]$ as $\mathcal{I}$.
Then construct $\mathbf{A}_1\in\mathbb{C}^{nL_1\times r}$  by choosing $r$ column vectors of $\mathbf{G}_{1}^{n}[\mathbf{V}_{2}^{n} \cdots \mathbf{V}_{K}^{n}]$ whose indexes are in $\mathcal{I}$ and construct $\mathbf{A}_2\in\mathbb{C}^{nL_1\times r}$ by choosing $r$ column vectors of $\boldsymbol{\Gamma}_{1}^{n}\mathbf{H}_{1}^{n}[\mathbf{V}_{2}^{n} \cdots \mathbf{V}_{K}^{n}]$ whose indexes are in $\mathcal{I}$.
Clearly, $\mathbf{A}_1$ is of full-rank.
There exist  $\binom{nL_{1}}{r}$ choices of constructing $r\times r$  submatrices from $\mathbf{A}_1$ (or $\mathbf{A}_2$) and the determinant of each of these submatrices can be expressed as a polynomial with respect to the entries of $\mathbf{G}_{1}^{c}$.

Suppose that all $r \times r$ submatrix of $\mathbf{A}_{2}$ of which determinant is a zero polynomial with respect to the entries of $\mathbf{G}_{1}^{c}$, i.e., a constant polynomial whose coefficients are all equal to zero.
In this case, $\mathbf{A}_{2}$ is not of full-rank regardless of the entries of $\mathbf{G}_{1}^{c}$.
Hence, any matrix constructed from $\mathbf{A}_{2}$ by substituting the entries of $\mathbf{G}_{1}^{c}$ with arbitrary values is not of full-rank either.
Let us now define  $\mathbf{A}_{3}$ constructed from $\mathbf{A}_{2}$ by substituting $\mathbf{G}_{1}^{c}$ with $\mathbf{P} \mathbf{G}_{1}$,  where $\mathbf{P} \in \mathbb{C}^{(N_{1}-L_{1})\times L_{1}}$, which is not of full-rank from the above argument.
Then, we can represent $\mathbf{A}_{3}$ as $\mathbf{A}_{3} = \mathbf{Q} \mathbf{A}_{1}$ for some matrix $\mathbf{Q} \in \mathbb{C}^{n L_{i} \times n L_{i}}$.
If all square submatrices of $\mathbf{P}$ are non-singular, $\mathbf{Q}$ becomes invertible so that $\mathbf{A}_{1}$ and $\mathbf{A}_{3}$ have the same rank. We can easily find such $\mathbf{P}$, for example,  Vandermonde matrix or Cauchy matrix \cite{bader2007petascale}.
However, from the fact that $\mathbf{A}_{3}$ is not of full-rank, the result that $\mathbf{A}_3$ and $\mathbf{A_1}$ have the same rank contradicts the assumption that $\mathbf{A}_{1}$ is of full-rank.
Consequently, there exists at least one $r \times r$ submatrix of $\mathbf{A}_{2}$ of which determinant is not a zero polynomial with respect to the entries of $\mathbf{G}_{1}^{c}$.

Then now consider some $r \times r$ submatrix of $\mathbf{A}_{2}$ of which determinant is not a zero polynomial with respect to the entries of $\mathbf{G}_{1}^{c}$.
Since the entries of $\mathbf{G}_{1}^{c}$ are i.i.d drawn from a  continuous distribution, for given $\mathbf{G}_{1}$ and $[\mathbf{V}_{2}^{n} \cdots \mathbf{V}_{K}^{n}]$, the determinant of  the considered submatrix of $\mathbf{A}_{2}$  is non-zero almost surely. 
Hence, $\mathbf{A}_{2}$ is of full-rank almost surely. 
Since $\mathbf{A}_{2}$ is a submatrix of $\boldsymbol{\Gamma}_{1}^{n}\mathbf{H}_{1}^{n}[\mathbf{V}_{2}^{n} \cdots \mathbf{V}_{K}^{n}]$,
 the rank of $\boldsymbol{\Gamma}_{1}^{n}\mathbf{H}_{1}^{n}[\mathbf{V}_{2}^{n} \cdots \mathbf{V}_{K}^{n}]$ is grater than or equal to that of $\mathbf{G}_{1}^{n}[\mathbf{V}_{2}^{n} \cdots \mathbf{V}_{K}^{n}]$ almost surely, which complete the proof of Lemma \ref{lemma_sel}.

\subsection{Proof of Lemma \ref{lemma_equi}}
\label{proof_equi}

In order to prove Lemma \ref{lemma_equi}, we need the following lemmas. 
The first lemma comes from the submodularity property for rank of matrices \cite{Lashgari:14,Lovasz1983}.
\begin{lemma}[Lov$\acute{a}$sz]
\label{lemma_submod}
For matrices $\mathbf{A}$, $\mathbf{B}$, and $\mathbf{C}$ with the same number of rows,
\begin{align}
\operatorname{rank}[\mathbf{A} \mathbf{C}] - \operatorname{rank}[\mathbf{C}] \geq \operatorname{rank}[\mathbf{A}\mathbf{B} \mathbf{C}] - \operatorname{rank}[\mathbf{B}\mathbf{C}] \notag .
\end{align}
\end{lemma}
\begin{proof}
We refer to \cite{Lovasz1983} for the proof.
\end{proof}


The second lemma is one of the key properties for multantenna systems without CSIT, which means that there is no spatial preference in the received signal space without CSIT.
\begin{lemma}
\label{lemma_stat}
Let $\mathbf{A}_{i,m}$ be the submatrix consisting of arbitrary $m$ row vectors of $\mathbf{G}_{i}$ and $\mathbf{B}_{j,m}$ be the submatrix consisting of arbitrary $m$ row vectors of $\mathbf{G}_{j}$.  
Then, for all $i,j,k \in \mathcal{K}$, the following property holds almost surely:
\begin{align}
\operatorname{rank}(\mathbf{A}_{i,m}^{n}[\mathbf{V}_{k}^{n}\cdots\mathbf{V}_{K}^{n}]) = \operatorname{rank}(\mathbf{B}^n_{j,m}[\mathbf{V}_{k}^{n}\cdots\mathbf{V}_{K}^{n}]) \notag
\end{align}
where $\mathbf{A}_{i,m}^{n}= \mathbf{I}_{n} \otimes \mathbf{A}_{i,m} $ and $\mathbf{B}_{j,m}^{n}=\mathbf{I}_{n} \otimes \mathbf{B}_{j,m} $.\footnote{The maximum value of $m$ depends on $i$ and $j$, see the definition of $\{\mathbf{G}_l\}_{l\in\mathcal{K}}$.}
\end{lemma}
\begin{IEEEproof}
For $i = j$, it can be straightforwardly  derived from the proof in Lemma \ref{lemma_sel}. Hence, we assume $i \neq j$ in the rest of the proof.
We first prove that 
\begin{align} \label{eq:rank_inequal}
\operatorname{rank}(\mathbf{A}_{i,m}^{n}[\mathbf{V}_{k}^{n}\cdots\mathbf{V}_{K}^{n}]) 
\overset{a.s.}{\leq} 
\operatorname{rank}(\mathbf{B}_{j,m}^{n}[\mathbf{V}_{k}^{n}\cdots\mathbf{V}_{K}^{n}])
\end{align}
for given realizations of $\mathbf{A}_{i,m}$ and $[\mathbf{V}_{k}^{n} \cdots \mathbf{V}_{K}^{n}]$.
Note that, for given $\mathbf{A}_{i,m}$ and $[\mathbf{V}_{k}^{n} \cdots \mathbf{V}_{K}^{n}]$, $\mathbf{A}_{i,m}^{n}[\mathbf{V}_{k}^{n}\cdots\mathbf{V}_{K}^{n}]$ is deterministic but $\mathbf{B}_{j,m}^{n}[\mathbf{V}_{k}^{n}\cdots\mathbf{V}_{K}^{n}]$ is random induced by $\mathbf{B}_{j,m}$.

If $\mathbf{A}_{i,m}^{n}[\mathbf{V}_{k}^{n}\cdots\mathbf{V}_{K}^{n}]=\mathbf{0}$, then \eqref{eq:rank_inequal} trivially holds.
Then now consider the case where $\mathbf{A}_{i,m}^{n}[\mathbf{V}_{k}^{n}\cdots\mathbf{V}_{K}^{n}]\neq \mathbf{0}$.
For convenience, denote $\operatorname{rank}(\mathbf{A}_{i,m}^{n}[\mathbf{V}_{k}^{n}\cdots\mathbf{V}_{K}^{n}]) = r \geq 1$.
Define the set of column indices consisting of $r$  linearly independent columns of $\mathbf{A}_{i,m}^{n}[\mathbf{V}_{k}^{n}\cdots\mathbf{V}_{K}^{n}]$ as $\mathcal{I}$.
Then construct $\mathbf{C}_1\in\mathbb{C}^{nm\times r}$  by choosing $r$ column vectors of $\mathbf{A}_{i,m}^{n}[\mathbf{V}_{k}^{n}\cdots\mathbf{V}_{K}^{n}]$ whose indexes are in $\mathcal{I}$
and $\mathbf{C}_2\in\mathbb{C}^{nm\times r}$  by choosing $r$ column vectors of $\mathbf{B}_{j,m}^{n}[\mathbf{V}_{k}^{n}\cdots\mathbf{V}_{K}^{n}]$ whose indexes are in $\mathcal{I}$.
Clearly, $\mathbf{C}_1$ is of full-rank.

There exist $\binom{nm}{r}$ $r \times r$ choices of constructing $r\times r$ submatrices from $\mathbf{C}_2$ and the determinant of each of these submatrices can be expressed as a polynomial with respect to the entries of $\mathbf{B}_{j,m}$.
Then from the same argument in the proof of Lemma \ref{lemma_sel}, we can show that  there exists at least one $r \times r$ submatrix of $\mathbf{C}_{2}$ of which determinant is not a zero polynomial with respect to the entries of $\mathbf{B}_{j,m}$.
Now consider one of such $r\times r$ submatrices of $\mathbf{C}_2$.
Since the entries of $\mathbf{B}_{j,m}$ are i.i.d drawn from a continuous distribution, for given $\mathbf{A}_{i,m}$ and $[\mathbf{V}_{k}^{n} \cdots \mathbf{V}_{K}^{n}]$, its determinant is non-zero almost surely. 
Hence, $\mathbf{C}_{2}$ is of full-rank almost surely and, as a result, \eqref{eq:rank_inequal} holds.
Similarly, we can also prove $\operatorname{rank}(\mathbf{A}_{i,m}^{n}[\mathbf{V}_{k}^{n}\cdots\mathbf{V}_{K}^{n}]) 
\overset{a.s.}{\geq} 
\operatorname{rank}(\mathbf{B}_{j,m}^{n}[\mathbf{V}_{k}^{n}\cdots\mathbf{V}_{K}^{n}])$.
In conclusion, Lemma \ref{lemma_stat} holds.
\end{IEEEproof}

We are now ready to prove Lemma \ref{lemma_equi}. 
Let us define 
$\mathbf{Z}^{i,j} = (\mathbf{G}_{i}^{n}[\mathbf{V}_{j}^{n}\cdots\mathbf{V}_{K}^{n}])^{T}$ and
$\mathbf{Z}_{k}^{i,j} = (\mathbf{g}_{i,k}^{n}[\mathbf{V}_{j}^{n}\cdots\mathbf{V}_{K}^{n}])^{T}$,
where $\mathbf{g}_{i,k}$ is the $k$th row vector of $\mathbf{G}_{i}$ and $\mathbf{g}_{i,k}^{n}= \mathbf{I}_{n} \otimes \mathbf{g}_{i,k}$. 
Then, for $i=2,\cdots,K$,
\begin{subequations}
\label{eq_equi4}
\begin{align}
\operatorname{rank}\left( \mathbf{Z}^{i-1,i} \right) &  \overset{a.s.}{=} \operatorname{rank} \left [ \mathbf{Z}_{\Delta_{i-1}}^{i-1,i} \cdots \mathbf{Z}_{1}^{i-1,i} \right]\label{eq_app0} \\
 &  = \operatorname{rank} \left ( \mathbf{Z}_{1}^{i-1,i} \right) 
+ \sum\limits_{k=2}^{\Delta_{i-1}} 
\left(
\operatorname{rank} \left [ \mathbf{Z}_{k}^{i-1,i} \cdots \mathbf{Z}_{1}^{i-1,i} \right] 
- \operatorname{rank} \left [ \mathbf{Z}_{k-1}^{i-1,i} \cdots \mathbf{Z}_{1}^{i-1,i} \right]
\right) \nonumber
\\
& \overset{a.s.}{=} \operatorname{rank} \left ( \mathbf{Z}_{\Delta_{i}-1}^{i-1,i} \right) 
+ \sum\limits_{k=2}^{\Delta_{i-1}} 
\left(
\operatorname{rank} \left [ \mathbf{Z}_{\Delta_{i-1}}^{i-1,i} \mathbf{Z}_{k-1}^{i-1,i}\cdots \mathbf{Z}_{1}^{i-1,i} \right] 
- \operatorname{rank} \left [ \mathbf{Z}_{k-1}^{i-1,i} \cdots \mathbf{Z}_{1}^{i-1,i} \right]
\right) 
\label{eq_app1} \\
& \geq   \sum\limits_{k=1}^{\Delta_{i-1}} 
\left(
\operatorname{rank} \left [ \mathbf{Z}_{\Delta_{i-1}}^{i-1,i} \cdots \mathbf{Z}_{1}^{i,i} \right] 
- \operatorname{rank} \left [ \mathbf{Z}_{\Delta_{i-1}-1}^{i-1,i} \cdots \mathbf{Z}_{1}^{i-1,i} \right]
\right) 
\label{eq_app2} \\
& \overset{a.s.}{=} \Delta_{i-1} \left(\operatorname{rank} \left [ \mathbf{Z}_{\Delta_{i-1}}^{i,i} \cdots \mathbf{Z}_{1}^{i,i} \right] 
- \operatorname{rank} \left [ \mathbf{Z}_{\Delta_{i-1}-1}^{i,i} \cdots \mathbf{Z}_{1}^{i,i} \right] 
\right) 
\label{eq_app3} \\
& \overset{a.s.}{=} \frac{\Delta_{i-1}}{(\Delta_{i} - \Delta_{i-1})}  
\sum\limits_{k=1}^{\Delta_{i} - \Delta_{i-1}} 
\left (
\operatorname{rank} \left [ \mathbf{Z}_{\Delta_{i-1}+k}^{i,i} \mathbf{Z}_{\Delta_{i-1}-1}^{i,i} \cdots \mathbf{Z}_{1}^{i,i} \right] 
- \operatorname{rank} \left [ \mathbf{Z}_{\Delta_{i-1}-1}^{i,i} \cdots \mathbf{Z}_{1}^{i,i} \right] 
\right ) 
\label{eq_app5}\\
& \geq \frac{\Delta_{i-1}}{(\Delta_{i} - \Delta_{i-1})}  \sum\limits_{k=1}^{\Delta_{i} - \Delta_{i-1}} 
\left(
\operatorname{rank} \left [ \mathbf{Z}_{\Delta_{i-1}+k}^{i,i} \cdots \mathbf{Z}_{1}^{i,i} \right] 
- \operatorname{rank} \left [ \mathbf{Z}_{\Delta_{i-1}+k-1}^{i,i} \cdots \mathbf{Z}_{1}^{i,i} \right] 
\right) 
\label{eq_app4} \\
& = \frac{\Delta_{i-1}}{(\Delta_{i} - \Delta_{i-1})}  
\left (
\operatorname{rank} \left [ \mathbf{Z}_{\Delta_{i}}^{i,i} \cdots \mathbf{Z}_{1}^{i,i} \right] 
- \operatorname{rank} \left [ \mathbf{Z}_{\Delta_{i-1}}^{i,i} \cdots \mathbf{Z}_{1}^{i,i} \right] 
\right) 
\notag \\
& \overset{a.s.}{=} \frac{\Delta_{i-1}}{(\Delta_{i} - \Delta_{i-1})}  
\left( 
\operatorname{rank} \left [ \mathbf{Z}_{\Delta_{i}}^{i,i} \cdots \mathbf{Z}_{1}^{i,i} \right] 
- \operatorname{rank} \left [ \mathbf{Z}_{\Delta_{i-1}}^{i-1,i} \cdots \mathbf{Z}_{1}^{i-1,i} \right] 
\right) 
\label{eq_app6} \\
& \overset{a.s.}{=} \frac{\Delta_{i-1}}{(\Delta_{i} - \Delta_{i-1})} \left(\operatorname{rank} \left ( \mathbf{Z}^{i,i} \right ) - \operatorname{rank} \left (  \mathbf{Z}^{i-1,i} \right) \right). \label{eq_app7} 
\end{align}
\end{subequations}
Here \eqref{eq_app0} holds since $[ \mathbf{Z}_{\Delta_{i-1}}^{i-1,i} \cdots \mathbf{Z}_{1}^{i-1,i}]$ is a submatrix of $\mathbf{Z}^{i-1,i}$ and 
$\mathbf{Z}^{i-1,i} =  [ \mathbf{Z}_{\Delta_{i-1}}^{i-1,i} \cdots \mathbf{Z}_{1}^{i-1,i}] \mathbf{F} $ almost surely for a matrix $\mathbf{F}$ such that
\begin{align}
\mathbf{F} = 
\left \{
\begin{array}{cc}
\mathbf{I}_{n} \otimes 
\left [  
\mathbf{I}_{\Delta_{i-1}}^{T} \  \left ((\left[\mathbf{G}_{i-1}\right]_{\Delta_{i-1}}^{T})^{-1}  (\left[\mathbf{G}_{i-1}\right]_{\Delta_{i-1}}^{c})^{T}  \right )^{T} \right ]^{T} & \textrm{if } \mathbf{G}_{i-1}  \textrm{ is a tall matrix}, \\
 \mathbf{I}_{n\Delta_{i-1}} & \textrm{otherwise} \\
\end{array}
\right.
\end{align}
where $[\mathbf{G}_{i-1}]_{\Delta_{i-1}}$ is the leading principal minor of $\mathbf{G}_{i-1}$ of order $\Delta_{i-1}$, $[\mathbf{G}_{i-1}]_{\Delta_{i-1}}^{c}$ is the remainder part of $\mathbf{G}_{i-1}$ except $[\mathbf{G}_{i-1}]_{\Delta_{i-1}}$. 
Lemma \ref{lemma_stat} is used for \eqref{eq_app1}, \eqref{eq_app3},  \eqref{eq_app5}, and \eqref{eq_app6} and Lemma \ref{lemma_submod} is used for \eqref{eq_app2} and \eqref{eq_app4}.
Also, \eqref{eq_app7}  follows since $\operatorname{rank} \left [ \mathbf{Z}_{\Delta_{i}}^{i,i} \cdots \mathbf{Z}_{1}^{i,i} \right]\overset{a.s.}{=}\operatorname{rank} \left ( \mathbf{Z}^{i,i} \right )$ and $\operatorname{rank} \left [ \mathbf{Z}_{\Delta_{i-1}}^{i-1,i} \cdots \mathbf{Z}_{1}^{i-1,i} \right]\overset{a.s.}{=}\operatorname{rank} \left ( \mathbf{Z}^{i-1,i} \right )$.
From the fact that $\operatorname{rank} \mathbf{A} = \operatorname{rank} \mathbf{A}^{T}$ for a matrix $\mathbf{A}$ whose elements are complex numbers \cite{horn2012matrix}, \eqref{eq_equi4} becomes
\begin{align}
\frac{1}{\Delta_{i-1}}\operatorname{rank} \left (\mathbf{G}_{i-1}^{n}[\mathbf{V}_{i}^{n} \cdots \mathbf{V}_{K}^{n}] \right )  \overset{a.s.}{\geq} \frac{1}{\Delta_{i}}\operatorname{rank} \left (\mathbf{G}_{i}^{n}[\mathbf{V}_{i}^{n} \cdots \mathbf{V}_{K}^{n}] \right ) \notag .
\end{align}
Then, for $i=2,\cdots,K$, we have
\begin{subequations}
\begin{align}
\frac{1}{\Delta_{i-1}}\operatorname{rank} \left (\mathbf{G}_{i-1}^{n}[\mathbf{V}_{i}^{n} \cdots \mathbf{V}_{K}^{n}] \right ) 
& \overset{a.s.}{\geq} \frac{1}{\Delta_{i}}\operatorname{rank} \left (\mathbf{G}_{i}^{n}[\mathbf{V}_{i}^{n} \cdots \mathbf{V}_{K}^{n}] \right ) \notag \\
&= \frac{1}{\Delta_{i}} \left( \operatorname{rank} (\mathbf{G}_{i}^{n}[\mathbf{V}_{i+1}^{n} \cdots \mathbf{V}_{K}^{n}]  ) + \operatorname{dim}(\operatorname{Proj}_{(\mathcal{I}_{i}')^{c}}\mathcal{R}(\mathbf{G}_{i}^{n}\mathbf{V}_{i}^{n}))  \right)\label{eq_equi1} \\
&\geq \frac{1}{\Delta_{i}}  \left( \operatorname{rank} \left (\mathbf{G}_{i}^{n}[\mathbf{V}_{i+1}^{n} \cdots \mathbf{V}_{K}^{n}]\right ) + \operatorname{dim}(\operatorname{Proj}_{\mathcal{I}_{i}^{c}}\mathcal{R}(\mathbf{G}_{i}^{n}\mathbf{V}_{i}^{n}))   \right)\label{eq_equi2}
\end{align}
\end{subequations}
where $\mathcal{I'}_{i} = \mathcal{R}(\mathbf{G}_{i}^{n}[\mathbf{V}_{i+1}^{n} \cdots \mathbf{V}_{K}^{n}])$.
Here \eqref{eq_equi1} follows from Lemma \ref{lemma_fano} and \eqref{eq_equi2} follows since $\mathcal{I}_{i}'  \subseteq \mathcal{I}_{i}$, which is given by $\mathcal{I}_{i} = \mathcal{R}( \mathbf{G}_{i}^{n}[\mathbf{V}_{1}^{n} \cdots \mathbf{V}_{i-1}^{n},\mathbf{V}_{i+1}^{n},\cdots,\mathbf{V}_{K}^{n}] )$ from Definition \ref{def:LDoF}.
Therefore \eqref{eq: relation_1} holds. In the same manner, we can proof \eqref{eq: relation_2}, which completes the proof of Lemma \ref{lemma_equi}.

\subsection{Proof of Lemma \ref{lemma_upper}}
\label{proof_upper}
Let us assume that $M > L_{\max}$ and $N_{k} > L_{\max}$ for some $k \in \mathcal{K}$, i.e., $\Lambda \neq \emptyset$.
Let $\Lambda = \{1,\cdots,|\Lambda| \}$ without loss of generality and
$\mathbf{A} = [\mathbf{A}_{1} \mathbf{A}_{2}]$ such that
\begin{align}
\label{eq_A1}
\begin{array}{l}
\mathbf{A}_{1} = 
\left[ 
\begin{array}{c}
\begin{array}{ccccc}
\frac{1}{L_{1}} & \frac{1}{T_{2}} & \frac{1}{T_{3}} & \cdots & \frac{1}{T_{|\Lambda|}} \\
\frac{1}{T_{1}} & \frac{1}{L_{2}} & \frac{1}{T_{3}} & \cdots & \frac{1}{T_{|\Lambda|}} \\
\vdots & \vdots & \vdots & \ddots & \vdots  \\
\frac{1}{T_{1}} & \frac{1}{T_{2}} & \frac{1}{T_{3}} & \cdots & \frac{1}{L_{|\Lambda|}} \\
\end{array}
\end{array}
\right]
\end{array}
\end{align}
and $\mathbf{A}_{2} = \frac{1}{L_{\max}} \mathbf{1}_{|\Lambda|\times (K-|\Lambda|)}$.
Then the optimization problem in Lemma \ref{lemma:linear_program} is rewritten as 
\begin{gather}
\textrm{maximize} \sum\limits_{i=1}^{K} d_{i} \notag \\
\textrm{subject to} \ \mathbf{Ad} + \mathbf{x} = \mathbf{1}_{|\Lambda| \times 1} \label{eq:opt_constraint}, \\
										\ \ \ \ \ \  \mathbf{d} \geq \mathbf{0}, \mathbf{x} \geq \mathbf{0} \notag
\end{gather}
where $\mathbf{d} = [d_{1} \cdots d_{K}]^{T}$, $\mathbf{x} = [x_{1} \cdots x_{|\Lambda|}]^{T}$, see also \cite{chong2013introduction}.

The following lemma provides non-singularity of $\mathbf{A}_{1}$. 
\begin{lemma}
\label{det}
For the matrix $\mathbf{A}_{1}$ defined in \eqref{eq_A1}, the determinant of $\mathbf{A}_{1}$ is given by
\begin{align}
|\mathbf{A}_{1}| = \prod\limits_{k\in\Lambda}\frac{T_{k}-L_{k}}{T_{k}L_{k}}\left ( 1 + \sum\limits_{k\in\Lambda} \frac{L_{j}}{T_{j}-L_{j}} \right ) .
\end{align}
Consequently, since $T_{i} > L_{i}$ for $i \in \Lambda$, $\mathbf{A}_{1}$ is non-singular.
\begin{IEEEproof}
It can be easily verified by mathematical induction. 
\end{IEEEproof}
\end{lemma}

Notice the the optimal $\mathbf{d}$ should satisfy \eqref{eq:opt_constraint}. Then, 
subtracting $\mathbf{x}$  from both sides of \eqref{eq:opt_constraint} and
multiplying them by $\mathbf{1}_{1\times |\Lambda|}\mathbf{A}_{1}^{-1}$, which is possible from Lemma \ref{det},  we have
\begin{align} 
\sum\limits_{i=1}^{|\Lambda|} d_{i} + \frac{\mathbf{1}_{1\times |\Lambda|}\mathbf{A}_{1}^{-1}\mathbf{1}_{|\Lambda|\times 1}}{L_{\max}}\sum\limits_{i=|\Lambda|+1}^{K} d_{i} = \mathbf{1}_{1\times |\Lambda|}\mathbf{A}_{1}^{-1}\mathbf{1}_{|\Lambda|\times 1} - \mathbf{1}_{1\times |\Lambda|}\mathbf{A}_{1}^{-1}\mathbf{x}.
\label{eq_upper2}
\end{align}
Note that
\begin{align}
\mathbf{1}_{1\times |\Lambda|}\mathbf{A}_{1}^{-1}\mathbf{1}_{|\Lambda|\times 1} 
&=\sum_{i\in \Lambda} [\mathbf{1}_{1\times |\Lambda|}\mathbf{A}_{1}^{-1}]_{i}\nonumber\\
&=\sum_{i\in \Lambda}  \frac{|\mathbf{A}_{1}|_{T_{i}=L_{i}=1}}{|\mathbf{A}_{1}|}\nonumber\\
&=\sum_{i\in \Lambda}\frac{\prod\limits_{k\in\Lambda, k\neq i}\frac{T_{k}-L_{k}}{T_{k}L_{k}}}
			{\prod\limits_{k\in\Lambda}\frac{T_{k}-L_{k}}{T_{k}L_{k}}\left ( 1 + \sum\limits_{j\in\Lambda} \frac{L_{j}}{T_{j}-L_{j}} \right )}\nonumber\\
&=\frac{\sum\limits_{i \in \Lambda}\frac{T_{i}L_{i}}{T_{i}-L_{i}}}{1+\sum\limits_{i\in\Lambda}\frac{L_{i}}{T_{i}-L_{i}}} = \eta. \label{eq:eta}
\end{align}
Here the third and fourth equalities follow from Cramer's rule \cite[Lemma 176]{gockenbach2011finite} and Lemma \ref{det} respectively.
Substituting \eqref{eq:eta} into \eqref{eq_upper2}, we have
\begin{align} 
\sum\limits_{i=1}^{|\Lambda|} d_{i} + \frac{\eta}{L_{\max}}\sum\limits_{i=|\Lambda|+1}^{K} d_{i} = \eta - \mathbf{1}_{1\times |\Lambda|}\mathbf{A}_{1}^{-1}\mathbf{x}.
\label{eq_upper5}
\end{align}
Therefore, 
\begin{align} 
\sum\limits_{i=1}^{K} d_{i} & \leq \max(\eta,L_{\max}) \left (  \frac{1}{\eta}\sum\limits_{i=1}^{|\Lambda|} d_{i} + \frac{1}{L_{\max}}\sum\limits_{i=|\Lambda|+1}^{K} d_{i} \right ) \nonumber\\
& = \max(\eta,L_{\max}) \frac{1}{\eta} \left ( \eta - \mathbf{1}_{1\times |\Lambda|}\mathbf{A}_{1}^{-1}\mathbf{x} \right ) \label{eq_upper3}\\
&\leq \max(\eta,L_{\max}) \label{eq_upper4} 
\end{align}
where \eqref{eq_upper3} follows from \eqref{eq_upper5} and \eqref{eq_upper4} follows since $\mathbf{1}_{1\times |\Lambda|}\mathbf{A}_{1}^{-1}\mathbf{x} \geq 0$. 
In conclusion, Lemma \ref{lemma_upper} holds.

\subsection{Proof of Lemma \ref{lemma:inverse}}
\label{properties1}
We have
\begin{align}
g^{[i]}(f^{[i]}(l)) & =  ((l-1) \setminus \prod\limits_{p=1}^{i}S_{p})\prod\limits_{p=1}^{i}S_{p} + (l-1) | \prod\limits_{p=1}^{i-1}S_{p} + 
(((l-1) | \prod\limits_{p=1}^{i}S_{p})\setminus \prod\limits_{p=1}^{i-1}S_{p})\prod\limits_{p=1}^{i-1}S_{p} + 1 \notag \\
& = ((l-1) \setminus \prod\limits_{p=1}^{i}S_{p})\prod\limits_{p=1}^{i}S_{p}  + (l-1) | \prod\limits_{p=1}^{i}S_{p} + 1 \label{eq_inv} \\
& = l \notag .
\end{align}
Here \eqref{eq_inv} follows that $(l-1) | \prod_{p=1}^{i-1}S_{p} = ((l-1) | \prod_{p=1}^{i}S_{p})|\prod_{p=1}^{i-1}S_{p}$.
In conclusion, Lemma \ref{lemma:inverse} holds.

\subsection{Proof of Lemma \ref{lemma:const}}
\label{properties2}
We now prove that, for $i, i' \in \mathcal{K}$ where $i \neq i'$,
\begin{align}
f_{2}^{[i]}(g^{[i']}(j,k)) = 
\begin{cases}
((j-1)|\prod\limits_{p=1}^{i'-1}S_{p} ) \setminus \prod\limits_{p=1}^{i-1}S_{p} + 1 & \mbox{if } i < i', \\
((j-1)\setminus\prod\limits_{p=1}^{i'-1}S_{p})|\prod\limits_{p=i'+1}^{i}S_{p})\setminus \prod\limits_{p=i'+1}^{i-1}S_{p} + 1 & \mbox{if } i > i'.
\end{cases}
\label{eq:const_proof}
\end{align}
From the definition of $f^{[i]}$ and $g^{[i']}$, \eqref{eq:const_proof} holds trivially for $i < i'$. Hence, assume $i > i'$ in the rest of this section.

For easy representation of the proof, for $i > i'$, denote, 
\begin{align}
& a_0 = (j-1)\setminus\prod_{p=1}^{i'-1}S_{p} \notag \\ 
& a_1 = a_0 \setminus \prod_{p=i'+1}^{i}S_p, \ \  b_1 = a_0 | \prod_{p=i'+1}^{i}S_p \notag \\
& a_2 = b_1 \setminus \prod_{p=i'+1}^{i-1}S_p, \ \ b_2  = b_1 | \prod_{p=i'+1}^{i-1}S_p \notag
\end{align}
From the definition of $g^{[i']}$, the following relation holds for $j \in \mathcal{A}$, $k \in \mathcal{B}$, and $i>i'$:
\begin{align}
g^{[i']}(j,k) - 1 & = a_0 \prod_{p=1}^{i'}S_{p} +  c  \notag \\
									& = a_1 \prod_{p=1}^{i}S_{p} +  b_1\prod_{p=1}^{i'}S_{p} +  c \label{eq:const_proof2}
\end{align}
where $c = (j-1)|\prod_{p = 1}^{i'-1}S_{p} + (k-1)\prod_{p=1}^{i'-1}S_{p}$ and the second equality follows since $a_0 = a_1 \prod_{p=i'+1}^{i}S_p + b_1$.
From the fact that $b_1 \leq \prod_{p = i' +1}^{i}S_p - 1$ and $c < \prod_{p = 1}^{i'}S_p$, one can see that $b_1\prod_{p=1}^{i'}S_{p} +  c < \prod_{p=1}^{i}S_{p}$, which results from \eqref{eq:const_proof2} that
\begin{align}
(g^{[i']}(j,k) - 1)|\prod_{p = 1}^{i}S_p & = b_1\prod_{p=1}^{i'}S_{p} +  c \notag \\
																				 & = a_2\prod_{p=1}^{i-1}S_{p} + b_2\prod_{p=1}^{i'}S_{p} +  c \label{eq:const_proof3}
\end{align}
where the second equality follows that $b_1 = a_2 \prod_{p=i'+1}^{i-1}S_p + b_2$.
Since $b_2 \leq \prod_{p = i' +1}^{i-1}S_p - 1$ and $c < \prod_{p = 1}^{i'}S_p$, one can see that $b_2\prod_{p=1}^{i'}S_{p} +  c < \prod_{p=1}^{i-1}S_{p}$, which results from \eqref{eq:const_proof3} that
\begin{align}
f_{2}^{[i]}(g^{[i']}(j,k)) & = ((g^{[i']}(j,k) - 1)|\prod_{p = 1}^{i}S_p) \setminus \prod_{p=1}^{i-1}S_p + 1 \notag \\
													 & = a_2 + 1, \notag 
\end{align}
which completes the proof of Lemma \ref{lemma:const}.

\section{Blind IA for a Two-User Example} \label{appendix:example}
For better understanding of the proposed blind IA stated in Section \ref{achievability_thm1}, we provide a two-user example here. 
Consider the two-user MIMO BC with reconfigurable antennas defined in Section \ref{sec:system_model} where $M = N_{1} = N_{2} = 3$, $L_{1} = 1$, and $L_{2} = 2$.
From \eqref{def:ach_parameter}, $T_{1}=T_{2}=3$, $S_{2}=W_{1}=1$, and $S_{2}=U_{1}=U_{2}=U=W_{2}=W = 2$.

\subsubsection{Transmit beamforming design} In Step 1, user 1 needs two information vectors ($U_{1}W_{1}=2$) of which size is three ($T_{1}L_{1} = 3$)
and user 2 needs four information vectors ($U_{2}W_{2}=4$) of which size is six ($T_{2}L_{2} = 6$).
Let us denote the information vectors of user 1 as $\mathbf{s}_{1}^{[1]}$ and $\mathbf{s}_{2}^{[1]} \in \mathbb{C}^{3}$ 
and denote the information vectors of user 2 as $\mathbf{s}_{1}^{[2]}$, $\mathbf{s}_{2}^{[2]}$, $\mathbf{s}_{3}^{[2]}$, and $\mathbf{s}_{4}^{[2]} \in \mathbb{C}^{6}$. 
Then, from \eqref{eq:align_block}, the alignment block of user 1 $\mathbf{v}_{j}^{[1]}$ for $j=1,2$ is given by
\begin{align}
\mathbf{v}_{j}^{[1]}
=
\left[
\begin{array}{c}
\mathbf{v}_{j,1}^{[1]} \\ \hline
\mathbf{v}_{j,2}^{[1]} \\ \hline
\mathbf{v}_{j,3}^{[1]}
\end{array}
\right]
= 
\left[
\begin{array}{c}
(\mathbf{I}_{1} \otimes \mathbf{I}_{3} ) \mathbf{s}_{j}^{[1]} \\ \hline
(\mathbf{I}_{1} \otimes \mathbf{I}_{3} ) \mathbf{s}_{j}^{[1]} \\ \hline
(\mathbf{I}_{1} \otimes \mathbf{I}_{3} ) \mathbf{s}_{j}^{[1]}
\end{array}
\right]
=
\left[
\begin{array}{c}
\mathbf{I}_{3} \\ \hline
\mathbf{I}_{3} \\ \hline
\mathbf{I}_{3}
\end{array}
\right] \mathbf{s}_{j}^{[1]} \in \mathbb{C}^{9} \notag, \ j = 1, 2
\end{align}
and from \eqref{eq:align_block}, the alignment block of user 2 $\mathbf{v}_{j}^{[2]}$ for $j=1,2,3,4$ are given by
\begin{align}
&
\mathbf{v}_{j}^{[2]} 
=
\left[
\begin{array}{c}
\mathbf{v}_{j,1}^{[1]} \\ \hline
\mathbf{v}_{j,2}^{[1]} \\
\end{array}
\right]
= 
\left[
\begin{array}{c}
(\boldsymbol{\Phi} \otimes \mathbf{I}_{3}) \mathbf{s}_{j}^{[2]} \\ \hline
(\mathbf{I}_{2} \otimes \mathbf{I}_{3}) \mathbf{s}_{j}^{[2]} \\
\end{array}
\right]
=
\left[
\begin{array}{c}
\boldsymbol{\Phi} \otimes \mathbf{I}_{3} \\ \hline 
\mathbf{I}_{6} \\
\end{array}
\right] \mathbf{s}_{j}^{[2]} \in \mathbb{C}^{9}, \notag \ j = 1, 2, 3, 4
\end{align}
where $\boldsymbol{\Phi} = [\phi_{11} \ \phi_{12}] \in \mathbb{C}^{1 \times 2}$ is a random matrix of which entries are i.i.d. continuous random variables.

In Step 2, we construct one $(W_{1} = 1)$ alignment unit of user 1 and two $(W_{2} = 2)$ alignment units of user 2.
Then, from \eqref{eq:align_unit},  alignment unit of user 1 $\mathbf{u}_{1}^{[1]}$ is given by
\begin{align}
\mathbf{u}_{1}^{[1]}
= \left[
\begin{array}{c}
\mathbf{u}_{1,1}^{[1]} \\ \hline
\mathbf{u}_{1,2}^{[1]} \\ \hline
\mathbf{u}_{1,3}^{[1]} \\
\end{array}
\right ]
= \left[
\begin{array}{c}
\mathbf{v}_{1,1}^{[1]} \\
\mathbf{v}_{2,2}^{[1]} \\ \hline
\mathbf{v}_{1,2}^{[1]} \\
\mathbf{v}_{2,1}^{[1]} \\ \hline
\mathbf{v}_{1,3}^{[1]} \\
\mathbf{v}_{2,3}^{[1]} \\
\end{array}
\right ]
= 
\left[
\begin{array}{c}
\mathbf{I}_{6} \\ \hline
\mathbf{I}_{6} \\ \hline
\mathbf{I}_{6} \\
\end{array}
\right ]
\left [
\begin{array}{c}
\mathbf{s}_{1}^{[1]} \\
\mathbf{s}_{2}^{[1]}
\end{array}
\right ] \in \mathbb{C}^{18} \notag
\end{align}
and the alignment unit of user 2 $\mathbf{u}_{j}^{[2]}$ for $j=1,2$ is given by
\begin{align}
\mathbf{u}_{j}^{[2]}
=
\left[
\begin{array}{c}
\mathbf{u}_{j,1}^{[2]} \\ \hline
\mathbf{u}_{j,2}^{[2]} \\
\end{array}
\right ]
=
\left[
\begin{array}{c}
\mathbf{v}_{2j-1,1}^{[2]} \\
\mathbf{v}_{2j,1}^{[2]} \\ \hline
\mathbf{v}_{2j-1,2}^{[2]} \\
\mathbf{v}_{2j,2}^{[2]} \\
\end{array}
\right ]
= \left[
\begin{array}{cc}
\boldsymbol{\Phi} \otimes \mathbf{I}_{3} & \mathbf{0}_{3 \times 6} \\
\mathbf{0}_{3 \times 6} & \boldsymbol{\Phi} \otimes \mathbf{I}_{3} \\ \hline
\mathbf{I}_{6} & \mathbf{0}_{6}\\
\mathbf{0}_{6} & \mathbf{I}_{6}\\
\end{array}
\right ]
\left [
\begin{array}{c}
\mathbf{s}_{2j-1}^{[2]} \\ 
\mathbf{s}_{2j}^{[2]}
\end{array}
\right ] \in \mathbb{C}^{18}. \notag
\end{align}

In Step 3, we construct the transmit signal vector for each user.
For this case, $f^{[i]}$ for $i=1,2$ is given by
\begin{align}
\left [ f^{[1]}(1) \ f^{[1]}(2)  \right ] = \left [ (1,1) \ (1,2) \right ], \
\left [ f^{[2]}(1) \ f^{[2]}(2)  \right ] = \left [ (1,1) \ (2,1) \right ] \notag
\end{align}
Then, from \eqref{eq_x1x3}, we have 
\begin{align}
& \mathbf{x}_{1,1} =
\left [ (\mathbf{u}_{1,1}^{[1]})^{T} \ (\mathbf{u}_{1,2}^{[1]})^{T} \right ]^{T}
, \ 
\mathbf{x}_{1,2} = \mathbf{u}_{1,3}^{[1]},
\notag
\\
& \mathbf{x}_{2,1}
=
\left [ (\mathbf{u}_{1,1}^{[2]})^{T} \ (\mathbf{u}_{2,1}^{[2]})^{T} \right ]^{T}
, \ 
\mathbf{x}_{2,2} = 
\left [ (\mathbf{u}_{1,2}^{[2]})^{T} \ (\mathbf{u}_{2,2}^{[2]})^{T} \right ]^{T}.
 \notag
\end{align}
Subsequently, from \eqref{eq_x1x2}, the transmit signal vector for user 1 is given by
\begin{align}
\mathbf{x}_{1}
=\left [
\begin{array}{c}
\mathbf{x}_{1,1} \\ \hline
\mathbf{x}_{1,2} \\ \hline
\mathbf{0}_{24 \times 1} \\
\end{array}
\right]
= 
\left [
\begin{array}{c}
\mathbf{u}_{1,1}^{[1]} \\ 
\mathbf{u}_{1,2}^{[1]} \\ \hline
\mathbf{u}_{1,3}^{[1]} \\ \hline
\mathbf{0}_{24 \times 6} \\
\end{array}
\right ]
=\left [
\begin{array}{c}
\mathbf{I}_{6} \\
\mathbf{I}_{6} \\ \hline
\mathbf{I}_{6} \\ \hline
\mathbf{0}_{24 \times 6} \\
\end{array}
\right]
\left [
\begin{array}{c}
\mathbf{s}_{1}^{[1]} \\
\mathbf{s}_{2}^{[1]}
\end{array}
\right ] \in \mathbb{C}^{42} \notag
\end{align}
and the transmit signal vector for user 2 is given by
\begin{align}
\mathbf{x}_{2}
=\left [
\begin{array}{c}
\mathbf{x}_{2,1} \\ \hline
\mathbf{0}_{6 \times 1} \\ \hline
\mathbf{x}_{2,2} \\
\end{array}
\right]
= 
\left [
\begin{array}{c}
\mathbf{u}_{1,1}^{[2]} \\ 
\mathbf{u}_{2,1}^{[2]} \\ \hline
\mathbf{0}_{6 \times 1} \\ \hline
\mathbf{u}_{1,2}^{[2]} \\ 
\mathbf{u}_{2,2}^{[2]} \\
\end{array}
\right ]
 \notag =
\left [
\begin{array}{cccc}
\boldsymbol{\Phi} \otimes \mathbf{I}_{3} & \mathbf{0}_{3 \times 6} & \mathbf{0}_{3 \times 6} & \mathbf{0}_{3 \times 6}\\
\mathbf{0}_{3 \times 6} & \boldsymbol{\Phi} \otimes \mathbf{I}_{3} & \mathbf{0}_{3 \times 6} & \mathbf{0}_{3 \times 6}\\
\mathbf{0}_{3 \times 6} & \mathbf{0}_{3 \times 6}  & \boldsymbol{\Phi} \otimes \mathbf{I}_{3} & \mathbf{0}_{3 \times 6} \\
\mathbf{0}_{3 \times 6} & \mathbf{0}_{3 \times 6} & \mathbf{0}_{3 \times 6} & \boldsymbol{\Phi} \otimes \mathbf{I}_{3} \\ \hline
\mathbf{0}_{6} & \mathbf{0}_{6} & \mathbf{0}_{6} & \mathbf{0}_{6}
\\ \hline
\mathbf{I}_{6} & \mathbf{0}_{6} & \mathbf{0}_{6} & \mathbf{0}_{6}\\ 
\mathbf{0}_{6} & \mathbf{I}_{6} & \mathbf{0}_{6} & \mathbf{0}_{6}\\
\mathbf{0}_{6} & \mathbf{0}_{6} & \mathbf{I}_{6} & \mathbf{0}_{6}\\
\mathbf{0}_{6} & \mathbf{0}_{6} & \mathbf{0}_{6} & \mathbf{I}_{6}
\end{array}
\right ]
\left [
\begin{array}{c}
\mathbf{s}_{1}^{[2]} \\
\mathbf{s}_{2}^{[2]} \\
\mathbf{s}_{3}^{[2]} \\
\mathbf{s}_{4}^{[2]}
\end{array}
\right ] \in \mathbb{C}^{42}. \notag
\end{align}
In Step 4, the overall transmit signal vector is given by
\begin{align}
\mathbf{x}^{n} = \mathbf{x}_{1} + \mathbf{x}_{2} \notag
\end{align}
where $n = 14$. 

\subsubsection{Mode switching patterns at receivers}
From \eqref{eq:time_block1}, the time interval for transmitting block 1 is given by $1 \leq t \leq 4$.
For this case, user 1 has two selection patterns and user 2 has one selection pattern, in which the associated channel matrices are given by
\begin{align}
& \mathbf{H}_{1,j} = \mathbf{h}_{1,j} \in \mathbb{C}^{1 \times 3}, \ j = 1, 2, \label{eq:example_pattern} \\
& \mathbf{H}_{2,1} =
\left [\mathbf{h}_{2,1}^{T} \ \mathbf{h}_{2,2}^{T} \right ]^{T} \in \mathbb{C}^{2 \times 3} \notag
\end{align}
respectively.
As explained in Section \ref{achievability_thm1}, when block 1 is transmitted, each user chooses the selection pattern of which index is the same as that of the currently transmitted sub-unit of his transmit signal vector.
Since the indexes of the transmitted sub-unit of users 1 and 2 are 1, 1, 2, 2 and 1, 1, 1, 1 respectively for $1 \leq t \leq 4$, from \eqref{eq:y_i0},
the received signal vectors of users 1 and 2 during $1 \leq t \leq 4$ are given by
\begin{align} 
& \mathbf{y}_{1,0} = 
\left [
\begin{array}{c}
\begin{array}{cc}
\mathbf{H}_{1,1} & \mathbf{0}_{1 \times 3} \\
\mathbf{0}_{1 \times 3} & \mathbf{H}_{1,1} \\
\mathbf{H}_{1,2} & \mathbf{0}_{1 \times 3} \\
\mathbf{0}_{1 \times 3} & \mathbf{H}_{1,2} \\
\end{array}
\end{array}
\right ]
\left [
\begin{array}{c}
\mathbf{s}_{1}^{[1]} \\
\mathbf{s}_{2}^{[1]}
\end{array}
\right ] 
+
\left [
\begin{array}{cccc}
\boldsymbol{\Phi} \otimes \mathbf{H}_{1,1} & \mathbf{0}_{1 \times 6} & \mathbf{0}_{1 \times 6} & \mathbf{0}_{1 \times 6} \\
\mathbf{0}_{1 \times 6} & \boldsymbol{\Phi} \otimes \mathbf{H}_{1,1} & \mathbf{0}_{1 \times 6} & \mathbf{0}_{1 \times 6} \\
\mathbf{0}_{1 \times 6} & \mathbf{0}_{1 \times 6} & \boldsymbol{\Phi} \otimes \mathbf{H}_{1,2} & \mathbf{0}_{1 \times 6} \\
\mathbf{0}_{1 \times 6} & \mathbf{0}_{1 \times 6} & \mathbf{0}_{1 \times 6} & \boldsymbol{\Phi} \otimes \mathbf{H}_{1,2} \\
\end{array}
\right ]
\left [
\begin{array}{c}
\mathbf{s}_{1}^{[2]} \\
\mathbf{s}_{2}^{[2]} \\
\mathbf{s}_{3}^{[2]} \\
\mathbf{s}_{4}^{[2]}
\end{array}
\right ] , \label{eq:example_y10}
\\ &
\mathbf{y}_{2,0} = 
\left [
\begin{array}{c}
\begin{array}{cc}
\mathbf{H}_{2,1} & \mathbf{0}_{2 \times 3} \\
\mathbf{0}_{2 \times 3} & \mathbf{H}_{2,1} \\
\mathbf{H}_{2,1} & \mathbf{0}_{2 \times 3} \\
\mathbf{0}_{2 \times 3} & \mathbf{H}_{2,1} \\
\end{array}
\end{array}
\right ]
\left [
\begin{array}{c}
\mathbf{s}_{1}^{[1]} \\
\mathbf{s}_{2}^{[1]}
\end{array}
\right ] 
+
\left [
\begin{array}{cccc}
\boldsymbol{\Phi} \otimes \mathbf{H}_{2,1} & \mathbf{0}_{2 \times 6} & \mathbf{0}_{2 \times 6} & \mathbf{0}_{2 \times 6} \\
\mathbf{0}_{2 \times 6} & \boldsymbol{\Phi} \otimes \mathbf{H}_{2,1} & \mathbf{0}_{2 \times 6} & \mathbf{0}_{2 \times 6} \\
\mathbf{0}_{2 \times 6} & \mathbf{0}_{2 \times 6} & \boldsymbol{\Phi} \otimes \mathbf{H}_{2,1} & \mathbf{0}_{2 \times 6} \\
\mathbf{0}_{2 \times 6} & \mathbf{0}_{2 \times 6} & \mathbf{0}_{2 \times 6} & \boldsymbol{\Phi} \otimes \mathbf{H}_{2,1} \\
\end{array}
\right ]
\left [
\begin{array}{c}
\mathbf{s}_{1}^{[2]} \\
\mathbf{s}_{2}^{[2]} \\
\mathbf{s}_{3}^{[2]} \\
\mathbf{s}_{4}^{[2]}
\end{array}
\right ] \label{eq:example_y20}
\end{align}
respectively.

The time interval for transmitting block 2 is given by $5 \leq t \leq 14$. 
First, we consider the time interval for transmitting $\mathbf{x}_{1,2}$, given by $5 \leq t \leq 6$.
Note that this part corresponds to the desired signal part of block 2 for user 1 and the interference signal part of block 2 for  user 2.
Since $T_{1}|L_{1} = 0$, user 1 exploits one selection pattern repeatedly over the time interval $5 \leq t \leq 6$, in which the associated channel matrix is given by
\begin{align}
\mathbf{H}_{1,3} = \mathbf{h}_{1,3} \in \mathbb{C}^{1 \times 3} \notag .
\end{align}
On the other hands, user 2 exploits the same selection pattern used for receiving during block 1 over the time interval $5 \leq t \leq 6$, in which there is one selection pattern associated with $\mathbf{H}_{2,1}$ for this case.
Then, user 2 receives the transmit signal for $5 \leq t \leq 6$ using the selection pattern associated with $\mathbf{H}_{2,1}$. 
As a result, from \eqref{eq_rs3_desired} and \eqref{eq_rs3}, the received signal vectors of users 1 and 2 during $5 \leq t \leq 6$ are given by
\begin{align} 
\mathbf{y}_{1,1} = 
\left [
\begin{array}{c}
\begin{array}{cc}
\mathbf{H}_{1,3} & \mathbf{0}_{1 \times 3} \\
\mathbf{0}_{1 \times 3} & \mathbf{H}_{1,3} \\
\end{array}
\end{array}
\right ]
\left [
\begin{array}{c}
\mathbf{s}_{1}^{[1]} \\
\mathbf{s}_{2}^{[1]}
\end{array}
\right ], 
\
\mathbf{y}_{2,1} = 
\left [
\begin{array}{c}
\begin{array}{cc}
\mathbf{H}_{2,1} & \mathbf{0}_{2 \times 3} \\
\mathbf{0}_{2 \times 3} & \mathbf{H}_{2,1} \\
\end{array}
\end{array}
\right ]
\left [
\begin{array}{c}
\mathbf{s}_{1}^{[1]} \\
\mathbf{s}_{2}^{[1]}
\end{array}
\right ] \label{eq:example_x11}
\end{align}
respectively.
Next, we consider the time interval for transmitting $\mathbf{x}_{2,2}$, given by $7 \leq t \leq 14$.
This part corresponds to the desired signal part of block 2 for user 2 and the interference signal part of block 2 for user 1.
Since $T_{2}|L_{2} \neq 0$, user 2 exploits two ($L_{2}S_{2} = 2)$ selection patterns, which repeat four ($S_{1}(T_{1}-L_{1}) = 4$) times periodically over the time interval $7 \leq t \leq 14$. The associated channel matrices are given by
\begin{align}
\mathbf{H}_{1,2,1} = \left [ \mathbf{h}_{2,3}^{T} \ \mathbf{h}_{2,1}^{T} \right ]^{T} \in \mathbb{C}^{2 \times 3}, \ 
\mathbf{H}_{1,2,2} = \left [ \mathbf{h}_{2,3}^{T} \ \mathbf{h}_{2,2}^{T} \right ]^{T} \in \mathbb{C}^{2 \times 3}. \notag
\end{align}
On the other hands,  user 1 exploits the same selection pattern used for receiving block 1 over the time interval $7 \leq t \leq 14$, 
in which the associated channel matrices are given in \eqref{eq:example_pattern}.
When the interference signal part of block 2 is transmitted, user 2 chooses the selection pattern of which index is the same as that used to receive the first sub-unit  of the alignment unit to which the currently transmitted sub-unit belongs. 
One can see that the sub-units transmitted for $7 \leq t \leq 10$ and $11 \leq t \leq 14$ is $\mathbf{u}_{1,3}^{[2]}$ and $\mathbf{u}_{2,3}^{[2]}$ respectively 
and user 1 exploits the selection pattern associated with $\mathbf{H}_{1,1}$ to receive $\mathbf{u}_{1,1}^{[2]}$
and the selection pattern associated with $\mathbf{H}_{1,2}$ to receive $\mathbf{u}_{2,1}^{[2]}$ in block 1.
Hence,  user 1 exploits the selection pattern associated with $\mathbf{H}_{1,1}$ for $7 \leq t \leq 10$ and the selection pattern associated with $\mathbf{H}_{1,2}$ for $11 \leq t \leq 14$. As a result, from \eqref{eq_rs3_desired} and \eqref{eq_rs3}, the received signal vectors of user 1 and 2 during $7 \leq t \leq 14$ are given by
\begin{align}
&
\mathbf{y}_{1,2} = 
\left [
\begin{array}{cccccccc}
\mathbf{I}_{2} \otimes \mathbf{H}_{1,1} & \mathbf{0}_{2 \times 6} & \mathbf{0}_{2 \times 6} & \mathbf{0}_{2 \times 6}\\
\mathbf{0}_{2 \times 6} & \mathbf{I}_{2} \otimes \mathbf{H}_{1,1} & \mathbf{0}_{2 \times 6} & \mathbf{0}_{2 \times 6}\\
\mathbf{0}_{2 \times 6} & \mathbf{0}_{2 \times 6} & \mathbf{I}_{2} \otimes \mathbf{H}_{1,2} & \mathbf{0}_{2 \times 6}\\
\mathbf{0}_{2 \times 6} & \mathbf{0}_{2 \times 6} & \mathbf{0}_{2 \times 6} & \mathbf{I}_{2} \otimes \mathbf{H}_{1,2}\\
\end{array}
\right ]
\left [
\begin{array}{c}
\mathbf{s}_{1}^{[2]} \\
\mathbf{s}_{2}^{[2]} \\
\mathbf{s}_{3}^{[2]} \\
\mathbf{s}_{4}^{[2]}
\end{array}
\right ], \label{eq:example_y12}
\\ &
\mathbf{y}_{2,2} = 
\left [
\begin{array}{cccccccc}
\mathbf{H}_{2,2}' & \mathbf{0}_{4 \times 6} & \mathbf{0}_{4 \times 6} & \mathbf{0}_{4 \times 6} \\
\mathbf{0}_{4 \times 6} & \mathbf{H}_{2,2}' & \mathbf{0}_{4 \times 6} & \mathbf{0}_{4 \times 6} \\
\mathbf{0}_{4 \times 6} & \mathbf{0}_{4 \times 6} & \mathbf{H}_{2,2}' & \mathbf{0}_{4 \times 6} \\
\mathbf{0}_{4 \times 6} & \mathbf{0}_{4 \times 6} & \mathbf{0}_{4 \times 6} & \mathbf{H}_{2,2}' \\
\end{array}
\right ]
\left [
\begin{array}{c}
\mathbf{s}_{1}^{[2]} \\
\mathbf{s}_{2}^{[2]} \\
\mathbf{s}_{3}^{[2]} \\
\mathbf{s}_{4}^{[2]}
\end{array}
\right ] \label{eq:example_y22}
\end{align}
respectively, where $\mathbf{H}_{2,2}' = \operatorname{diag}(\mathbf{H}_{2,2,1},\mathbf{H}_{2,2,2}) \in \mathbb{C}^{4 \times 6}$.

\subsubsection{Interference cancellation and achievable LDoF}
From \eqref{eq_result}, after cancelling all interference vectors in $\mathbf{y}_{1,0}$, user 1 has
\begin{align}
\left [
\begin{array}{c}
\mathbf{y}_{1,0} - (\mathbf{I}_{4} \otimes \boldsymbol{\Phi}) \mathbf{y}_{1,2} \\ \hline
\mathbf{y}_{1,1} \\
\end{array}
\right ]
=
\left [
\begin{array}{cc}
\mathbf{H}_{1,1} & \mathbf{0}_{1 \times 3} \\
\mathbf{0}_{1 \times 3} & \mathbf{H}_{1,1} \\
\mathbf{H}_{1,2} & \mathbf{0}_{1 \times 3} \\
\mathbf{0}_{1 \times 3} & \mathbf{H}_{1,2} \\ \hline
\mathbf{H}_{1,3} & \mathbf{0}_{1 \times 3} \\
\mathbf{0}_{1 \times 3} & \mathbf{H}_{1,3} \\
\end{array}
\right ]
\left [
\begin{array}{c}
\mathbf{s}_{1}^{[1]} \\
\mathbf{s}_{2}^{[1]}
\end{array}
\right ] \label{eq_example1}
\end{align}
Sorting the rows in \eqref{eq_example1}, we have
\begin{align}
\left [
\begin{array}{c}
\begin{array}{cc}
\mathbf{H}_{1} & \mathbf{0}_{ 3} \\
\mathbf{0}_{3} & \mathbf{H}_{1} \\
\end{array}
\end{array}
\right ]
\left [
\begin{array}{c}
\mathbf{s}_{1}^{[1]} \\
\mathbf{s}_{2}^{[1]}
\end{array}
\right ] \label{eq_example2}
\end{align}
Obviously,  user 1 can obtain $\mathbf{s}_{1}^{[1]}$ and $\mathbf{s}_{2}^{[1]}$ from \eqref{eq_example2} almost surely.

Similarly, from \eqref{eq_result}, after cancelling all interference vectors in $\mathbf{y}_{2,0}$, user 2 has
\begin{align}
\left [
\begin{array}{c}
\mathbf{y}_{2,0} - \mathbf{1}_{2 \times 1} \otimes \mathbf{y}_{2,1} \\ \hline
\mathbf{y}_{2,2} \\
\end{array}
\right ]
=
\left [
\begin{array}{cccccccc}
\boldsymbol{\Phi} \otimes \mathbf{H}_{2,1} & \mathbf{0}_{2 \times 6} & \mathbf{0}_{2 \times 6} & \mathbf{0}_{2 \times 6} \\
\mathbf{0}_{2 \times 6} & \boldsymbol{\Phi} \otimes \mathbf{H}_{2,1} & \mathbf{0}_{2 \times 6} & \mathbf{0}_{2 \times 6} \\
\mathbf{0}_{2 \times 6} & \mathbf{0}_{2 \times 6} & \boldsymbol{\Phi} \otimes \mathbf{H}_{2,1} & \mathbf{0}_{2 \times 6} \\
\mathbf{0}_{2 \times 6} & \mathbf{0}_{2 \times 6} & \mathbf{0}_{2 \times 6} & \boldsymbol{\Phi} \otimes \mathbf{H}_{2,1} \\ \hline
\mathbf{H}_{2,2}' & \mathbf{0}_{4 \times 6} & \mathbf{0}_{4 \times 6} & \mathbf{0}_{4 \times 6} \\
\mathbf{0}_{4 \times 6} & \mathbf{H}_{2,2}' & \mathbf{0}_{4 \times 6} & \mathbf{0}_{4 \times 6} \\
\mathbf{0}_{4 \times 6} & \mathbf{0}_{4 \times 6} & \mathbf{H}_{2,2}' & \mathbf{0}_{4 \times 6} \\
\mathbf{0}_{4 \times 6} & \mathbf{0}_{4 \times 6} & \mathbf{0}_{4 \times 6} & \mathbf{H}_{2,2}' \\
\end{array}
\right ]
\left [
\begin{array}{c}
\mathbf{s}_{1}^{[2]} \\
\mathbf{s}_{2}^{[2]} \\
\mathbf{s}_{3}^{[2]} \\
\mathbf{s}_{4}^{[2]}
\end{array}
\right ] \label{eq_example3}
\end{align}
Then, \eqref{eq_example3} can be decomposed into four segments as in the following.
\begin{align}
\left [
\begin{array}{cc}
\boldsymbol{\Phi} \otimes \mathbf{H}_{2,1} \\
\mathbf{H}_{2,2}'\\
\end{array}
\right ]
\mathbf{s}_{i}^{[2]}
=
\left [
\begin{array}{cc}
\phi_{11} \mathbf{h}_{2,1} & \phi_{12} \mathbf{h}_{2,1} \\
\phi_{11} \mathbf{h}_{2,2} & \phi_{12} \mathbf{h}_{2,2} \\
\mathbf{h}_{2,3} & \mathbf{0}_{1 \times 3} \\
\mathbf{h}_{2,1} & \mathbf{0}_{1 \times 3} \\
\mathbf{0}_{1 \times 3} & \mathbf{h}_{2,3}  \\
\mathbf{0}_{1 \times 3} & \mathbf{h}_{2,2}  \\
\end{array}
\right ]
\mathbf{s}_{i}^{[2]}, & \ \ \  i = 1, 2, 3, 4 \label{eq_example4}
\end{align}
It can be easily shown that user 2 can obtain $\mathbf{s}_{i}^{[2]}$ for all $i$ from \eqref{eq_example4} almost surely.

As a result, the transmitter delivers $6$ information symbols to user 1 and $24$ information symbols to user 2 during $14$ time slots.
Consequently, the achievable sum LDoF is given by $\frac{15}{7}$.


\begin{thebibliography}{10}

\bibitem{Cadambe:08}
V.~R. Cadambe and S.~A. Jafar,
\newblock ``Interference alignment and degrees of freedom of the {$K$}-user
  interference channel,''
\newblock {\em {IEEE} Trans. Inf. Theory}, vol. 54, pp. 3425--3441, Aug. 2008.

\bibitem{Viveck2:09}
V.~R. Cadambe and S.~A. Jafar,
\newblock ``Interference alignment and the degrees of freedom of wireless {$X$}
  networks,''
\newblock {\em {IEEE} Trans. Inf. Theory}, vol. 55, pp. 3893--3908, Sep. 2009.

\bibitem{Viveck1:09}
V.~R. Cadambe and S.~A. Jafar,
\newblock ``Degrees of freedom of wireless networks with relays, feedback,
  cooperation, and full duplex operation,''
\newblock {\em {IEEE} Trans. Inf. Theory}, vol. 55, pp. 2334--2344, May 2009.

\bibitem{Tiangao:10}
T.~Gou and S.~A. Jafar,
\newblock ``Degrees of freedom of the {$K$} user {$M\times N$} {MIMO}
  interference channel,''
\newblock {\em {IEEE} Trans. Inf. Theory}, vol. 56, pp. 6040--6057, Dec. 2010.

\bibitem{Suh:11}
C.~Suh, M.~Ho, and D.~N.~C. Tse,
\newblock ``Downlink interference alignment,''
\newblock {\em {IEEE} Trans. Commun.}, vol. 59, pp. 2616--2626, Sep. 2011.

\bibitem{Tiangao:12}
T.~Gou, S.~A. Jafar, C.~Wang, S.-W. Jeon, and S.-Y. Chung,
\newblock ``Aligned interference neutralization and the degrees of freedom of
  the $2\times2\times2$ interference channel,''
\newblock {\em {IEEE} Trans. Inf. Theory}, vol. 58, pp. 4381--4395, Jul. 2012.

\bibitem{Jeon4:12}
S.-W. Jeon and M.~Gastpar,
\newblock ``A survey on interference networks: Interference alignment and
  neutralization,''
\newblock {\em Entropy}, vol. 14, pp. 1842--1863, Sep. 2012.

\bibitem{Jeon:14}
S.-W. Jeon and C.~Suh,
\newblock ``Degrees of freedom of uplink--downlink multiantenna cellular
  networks,''
\newblock in \emph{arXiv:cs.IT/1404.6012}, Apr. 2014.

\bibitem{Nazer11:09}
B.~Nazer, M.~Gastpar, S.~A. Jafar, and S.~Vishwanath,
\newblock ``Ergodic interference alignment,''
\newblock {\em {IEEE} Trans. Inf. Theory}, vol. 58, pp. 6355--6371, Oct. 2012.

\bibitem{Jeon5:13}
S.-W. Jeon and S.-Y. Chung,
\newblock ``Capacity of a class of linear binary field multisource relay
  networks,''
\newblock {\em {IEEE} Trans. Inf. Theory}, vol. 59, pp. 6405--6420, Oct. 2013.

\bibitem{Jeon2:11}
S.-W. Jeon, S.-Y. Chung, and S.~A. Jafar,
\newblock ``Degrees of freedom region of a class of multisource {G}aussian
  relay networks,''
\newblock {\em {IEEE} Trans. Inf. Theory}, vol. 57, pp. 3032--3044, May 2011.

\bibitem{Jeon2:14}
S.-W. Jeon, C.-Y. Wang, and M.~Gastpar,
\newblock ``Approximate ergodic capacity of a class of fading two-user two-hop
  networks,''
\newblock {\em {IEEE} Trans. Inf. Theory}, vol. 60, pp. 866--880, Feb. 2014.

\bibitem{Motahari:09}
A.~S. Motahari, S.~O. Gharan, and A.~K. Khandani,
\newblock ``Real interference alignment with real numbers,''
\newblock in \emph{arXiv:cs.IT/0908.1208}, 2009.

\bibitem{Motahari2:09}
A.~S. Motahari, S.~O. Gharan, M.~A. Maddah-Ali, and A.~K. Khandani,
\newblock ``Real interference alignment: {E}xploiting the potential of single
  antenna systems,''
\newblock in \emph{arXiv:cs.IT/0908.2282}, 2009.

\bibitem{Maddah-Ali:12}
M.~A. Maddah-Ali and D.~Tse,
\newblock ``Completely stale transmitter channel state information is still
  very useful,''
\newblock {\em {IEEE} Trans. Inf. Theory}, vol. 58, pp. 4418--4431, Jul. 2012.

\bibitem{Vaze:11}
C.~S. Vaze and M.~K. Varanasi,
\newblock ``The degrees of freedom region of the two-user {MIMO} broadcast
  channel with delayed {CSIT},''
\newblock in {\em Proc. IEEE Int. Symp. Information Theory (ISIT)}, St.
  Petersburg, Russia, Jul. 2011.

\bibitem{Abdoli:11}
M.~J. Abdoli, A.~Ghasemi, and A.~K. Khandani,
\newblock ``On the degrees of freedom of three-user {MIMO} broadcast channel
  with delayed {CSIT},''
\newblock in {\em Proc. IEEE Int. Symp. Information Theory (ISIT)}, St.
  Petersburg, Russia, Jul. 2011.

\bibitem{Vaze:12_dCSIT}
C.~S. Vaze and M.~K. Varanasi,
\newblock ``The degrees of freedom region and interference alignment for the
  {MIMO} interference channel with delayed {CSIT},''
\newblock {\em {IEEE} Trans. Inf. Theory}, vol. 58, pp. 4396--4417, Jul. 2012.

\bibitem{Abdoli:11_IC}
M.~J. Abdoli, A.~Ghasemi, and A.~K. Khandani,
\newblock ``On the degrees of freedom of {SISO} interference and {X} channels
  with delayed {CSIT},''
\newblock in {\em Proc. 49th Annu. Allerton Conf. Communication, Control, and
  Computing}, Monticello, IL, Sep. 2011.

\bibitem{Abdoli:13}
M.~J. Abdoli, A.~Ghasemi, and A.~K. Khandani,
\newblock ``On the degrees of freedom of {K}-user {SISO} interference and {X}
  channels with delayed {CSIT},''
\newblock {\em {IEEE} Trans. Inf. Theory}, vol. 59, pp. 6542--6561, Oct. 2013.

\bibitem{Maleki:12}
H.~Maleki, S.~A. Jafar, and S.~Shamai,
\newblock ``Retrospective interference alignment over interference networks,''
\newblock {\em {IEEE} J. Sel. Topics Signal Process.}, vol. 6, pp. 228--240,
  Jun. 2012.

\bibitem{Jafar:05}
S.~A. Jafar and A.~J. Goldsmith,
\newblock ``Isotropic fading vector broadcast channels: The scalar upper bound
  and loss in degrees of freedom,''
\newblock {\em {IEEE} Trans. Inf. Theory}, vol. 51, pp. 848--857, Mar. 2005.

\bibitem{Vaze:12}
C.~S. Vaze and M.~K. Varanasi,
\newblock ``The degree-of-freedom regions of {MIMO} broadcast, interference,
  and cognitive radio channels with no {CSIT},''
\newblock {\em {IEEE} Trans. Inf. Theory}, vol. 58, pp. 5354--5374, Aug. 2012.

\bibitem{Huang:12}
C.~Huang, S.~A. Jafar, S.~Shamai, and S.~Vishwanath,
\newblock ``On degrees of freedom region of {MIMO} networks without channel
  state information at transmitters,''
\newblock {\em {IEEE} Trans. Inf. Theory}, vol. 58, pp. 849--857, Feb. 2012.

\bibitem{Zhu:12}
Y.~Zhu and D.~Guo,
\newblock ``The degrees of freedom of isotropic {MIMO} interference channels
  without state information at the transmitters,''
\newblock {\em {IEEE} Trans. Inf. Theory}, vol. 58, pp. 341--352, Jan. 2012.

\bibitem{Vaze:12_IC}
C.~S. Vaze and M.~K. Varanasi,
\newblock ``A new outer bound via interference localization and the degrees of
  freedom regions of {MIMO} interference networks with no {CSIT},''
\newblock {\em {IEEE} Trans. Inf. Theory}, vol. 58, pp. 6853--6869, Nov. 2012.

\bibitem{Jafar:12}
S.~A. Jafar,
\newblock ``Blind interference alignment,''
\newblock {\em {IEEE} J. Sel. Topics Signal Process.}, vol. 6, pp. 216--227,
  Jun. 2012.

\bibitem{Zhou:12}
Q.~F. Zhou and Q.~T. Zhang,
\newblock ``On blind interference alignment over homogeneous block fading
  channels,''
\newblock {\em {IEEE} Commun. Lett.}, vol. 16, pp. 1432--1435, Sep. 2012.

\bibitem{Zhou:12_dio}
Q.~F. Zhou, Q.~T. Zhang, and F.~C.~M. Lau,
\newblock ``Diophantine approach to blind interference alignment of homogeneous
  {K}-user 2$\times$1 {MISO} broadcast channels,''
\newblock {\em {IEEE} J. Sel. Areas Commun.}, vol. 31, pp. 2141--2153, Oct.
  2013.

\bibitem{Zhou:12_Imp}
Q.~F. Zhou, Q.~T. Zhang, and F.~C.~M. Lau,
\newblock ``Blind interference alignment over homogeneous 3-user 2$\times$1
  broadcast channel,''
\newblock in {\em Proc. International Workshop on High Mobility Wireless
  Communications (HMWC)}, Shanghai, China, Nov. 2013.

\bibitem{Gou:10}
T.~Gou, C.~Wang, and S.~A. Jafar,
\newblock ``Aiming perfectly in the dark - blind interference alignment through
  staggered antenna switching,''
\newblock in {\em Proc. {IEEE} {GLOBECOM}}, Miami, FL, Dec. 2010.

\bibitem{Gou:11}
T.~Gou, C.~Wang, and S.~A. Jafar,
\newblock ``Aiming perfectly in the dark-blind interference alignment through
  staggered antenna switching,''
\newblock {\em {IEEE} Trans. Signal Process.}, vol. 59, pp. 2734--2744, Jun.
  2011.

\bibitem{Wang:10}
C.~Wang, T.~Gou, and S.~A. Jafar,
\newblock ``Interference alignment through staggered antenna switching for
  {MIMO} {BC} with no {CSIT},''
\newblock in {\em Proc. Asilomar Conf. Sign., Syst., Computers}, Pacific Grove,
  CA, Nov. 2010.

\bibitem{Wang:11}
C.~Wang, H.~C. Papadopoulos, S.~A. Ramprashad, and G.~Caire,
\newblock ``Design and operation of blind interference alignment in cellular
  and cluster-based systems,''
\newblock in {\em Proc. Information Theory and Applications Workshop (ITA)}, La
  Jolla, CA, Feb. 2011.

\bibitem{Wang:11_2}
C.~Wang, H.~C. Papadopoulos, S.~A. Ramprashad, and G.~Caire,
\newblock ``Improved blind interference alignment in a cellular environment
  using power allocation and cell-based clusters,''
\newblock in {\em Proc. IEEE International Conference on Communications (ICC)},
  Kyoto, Japan, Jun. 2011.

\bibitem{Lu:13}
Y.~Lu and W.~Zhang,
\newblock ``Blind interference alignment in the {K}-user {MISO} interference
  channel,''
\newblock in {\em Proc. {IEEE} {GLOBECOM}}, Atlanta, GA, Dec. 2013.

\bibitem{Wang:14}
C.~Wang,
\newblock ``Degrees of freedom characterization: The 3-user {SISO} interference
  channel with blind interference alignment,''
\newblock {\em {IEEE} Commun. Lett.}, vol. 18, pp. 757--760, May 2014.

\bibitem{Lu:14}
Y.~Lu, W.~Zhang, and K.~B. Letaief,
\newblock ``Blind interference alignment with diversity in {$K$}-user
  interference channels,''
\newblock {\em to appear in IEEE Trans. Commun.}

\bibitem{Lashgari:13}
S.~Lashgari, A.~S. Avestimehr, and S.~Changho,
\newblock ``A rank ratio inequality and the linear degrees of freedom of
  {X}-channel with delayed {CSIT},''
\newblock in {\em Proc. 51st Annu. Allerton Conf. Communication, Control, and
  Computing}, Monticello, IL, Oct. 2013.

\bibitem{Lashgari:14}
S.~Lashgari, A.~S. Avestimehr, and C.~Suh,
\newblock ``Linear degrees of freedom of the {$X$}-channel with delayed
  {CSIT},''
\newblock {\em {IEEE} Trans. Inf. Theory}, vol. 60, pp. 2180--2189, Apr. 2014.

\bibitem{Kao:14}
D.~T.~H. Kao and A.~S. Avestimehr,
\newblock ``Linear degrees of freedom of the {MIMO} {X}-channel with delayed
  {CSIT},''
\newblock in {\em Proc. IEEE Int. Symp. Information Theory (ISIT)}, Honolulu,
  HI, Jun. 2014.

\bibitem{horn1991topics}
R.~A. Horn and C.~R. Johnson,
\newblock {\em Topics in Matrix Analysis},
\newblock Cambridge University Press, 1991.

\bibitem{chong2013introduction}
E.~K.~P. Chong and S.~H. Zak,
\newblock {\em An Introduction to Optimization},
\newblock Wiley Series in Discrete Mathematics and Optimization. Wiley, 2013.

\bibitem{bader2007petascale}
D.~A. Bader,
\newblock {\em Petascale Computing: Algorithms and Applications},
\newblock Chapman \& Hall/CRC Computational Science. Taylor \& Francis, 2007.

\bibitem{Lovasz1983}
L.~Lov$\acute{a}$sz,
\newblock ``Submodular functions and convexity,''
\newblock in {\em Mathematical Programming The State of the Art}. Springer
  Berlin Heidelberg, 1983.

\bibitem{horn2012matrix}
R.~A. Horn and C.~R. Johnson,
\newblock {\em Matrix Analysis},
\newblock Matrix Analysis. Cambridge University Press, 2012.

\bibitem{gockenbach2011finite}
M.~S. Gockenbach,
\newblock {\em Finite-Dimensional Linear Algebra},
\newblock Discrete Mathematics and Its Applications. Taylor \& Francis, 2011.

\end{thebibliography}
\end{document}